\documentclass[final]{IEEEtran}
\usepackage{color}
\usepackage{csquotes}
\usepackage{enumitem}
\usepackage{float}

\floatstyle{ruled}
\newfloat{algorithm}{tbp}{loa}
\providecommand{\algorithmname}{Algorithm}
\floatname{algorithm}{\protect\algorithmname}
\usepackage{hyperref}
\usepackage{amsopn} 

\DeclareMathOperator*{\argmax}{arg\,max}
\DeclareMathOperator*{\argmin}{arg\,min}

\makeatletter
\newcommand{\manuallabel}[2]{\def\@currentlabel{#2}\label{#1}}
\makeatother
\usepackage{amsmath}
\usepackage{amssymb}
\usepackage{amsthm}
\usepackage{bbm}
\usepackage{tikz}
\usepackage{mathrsfs}
\usepackage{pgfplots}
\pgfplotsset{compat=1.14}
\usetikzlibrary{arrows}
\usetikzlibrary{arrows.meta}

\newtheorem{theorem}{Theorem}
\newtheorem{lemma}[theorem]{Lemma}

\newtheorem{dfn}[theorem]{Definition}

\title{Multi-Bit Relaying over a Tandem of Channels}

\author{Yan Hao Ling and Jonathan Scarlett\thanks{The authors are with the  Department of Computer Science, National University of Singapore (NUS). Jonathan Scarlett is also with the Department of Mathematics, NUS, and the Institute of Data Science, NUS.  (e-mail: \url{lingyh@nus.edu.sg};  \url{scarlett@comp.nus.edu.sg}).} \thanks{This work was supported by the Singapore National Research Foundation (NRF) under grant number R-252-000-A74-281.}}

\newcommand{\edag}{\mathcal{E}^{\dagger}}
\newcommand{\p}{\mathbb{P}}
\newcommand{\e}{\mathbb{E}}
\newcommand{\eratetwo}{E^{(2)}}
\newcommand{\erateone}{E^{(1)}}

\newcommand{\calA}{\mathcal{A}}
\newcommand{\calB}{\mathcal{B}}

\newcommand{\calX}{\mathcal{X}}
\newcommand{\calY}{\mathcal{Y}}
\newcommand{\calC}{\mathcal{C}}
\newcommand{\calP}{\mathcal{P}}

\newcommand{\ratetwo}{\mathcal{E}^{(2)}}
\newcommand{\rateone}{\mathcal{E}^{(1)}}
\newcommand{\db}{d_{\rm B}}
\newcommand{\dbmin}{d^{\min}_{\rm B}}
\newcommand{\dc}{d_{\rm C}}
\newcommand{\dcmin}{d^{\min}_{\rm C}}
\newcommand{\altrate}{\mathcal{E}'}
\newcommand{\otherrate}{\mathcal{E}^{\dagger}}
\newcommand{\smm}{s_{m,m'}}
\newcommand{\empdist}{\hat{p}}
\newcommand{\erasure}{\mathsf{e}}

\begin{document}
    \maketitle
    
    \begin{abstract}
        We study error exponents for the problem of relaying a message over a tandem of two channels sharing the same transition law, in particular moving beyond the 1-bit setting studied in recent related works.  Our main results show that the 1-hop and 2-hop exponents coincide in both of the following settings: (i) the number of messages is fixed, and the channel law satisfies a condition called pairwise reversibility, or (ii) the channel is arbitrary, and a zero-rate limit is taken from above.  In addition, we provide various extensions of our results that relax the assumptions of pairwise reversibility and/or the two channels having identical transition laws, and we provide an example for which the 2-hop exponent is strictly below the 1-hop exponent.
    \end{abstract}
    
    \section{Introduction} 
    
    The relay channel is a fundamental building block of network information theory, and has many variations providing unique challenges and open problems.  In this work, we build on a recent line of works studying error exponents for transmitting a \emph{single bit} over a \emph{tandem of channels}, which was introduced by Huleihel, Polyanskiy, and Shayevitz \cite{onebit}, as well as Jog and Loh using different motivation/terminology based on teaching and learning in multi-agent problems \cite{jog2020teaching}.  
    
    In this 1-bit setting, we showed in \cite{teachlearn} that the 1-hop and 2-hop exponents coincide whenever the two channels have the same transition law (and also in certain other cases), which confirmed a conjecture from  \cite{onebit} inspired by the \emph{information velocity} (many-hop relaying) problem.  We provide further details and outline other related works in Section \ref{sec:existing}.

    In this paper, we study the natural extension of the preceding problem to the multi-bit setting.  We provide broad conditions under which the 1-hop and 2-hop exponents match, both in the case of a fixed finite number of messages, and in the case of a positive rate approaching zero from above.  The multi-bit setting comes with a variety of additional challenges that will become evident throughout the paper.

    \subsection{Problem Setup}

We first formalize the model, which is depicted in Figure \ref{fig:setup}. There are three agents: an encoder, relay, and decoder.  We focus on the case that the ``encoder $\to$ relay'' channel and the ``relay $\to$ decoder'' channel are the same (and independent of one another),\footnote{See Section \ref{sec:diff_chan} for the case that the two channels have different transition laws.} according to a discrete memoryless law $P$. 
 The unknown message of interest is random variable $\Theta$ drawn uniformly from $\{1, \ldots, M\}$.

At time step $i \in \{1,\dotsc, n\}$, the following occurs (simultaneously):
    \begin{itemize} 
        \item The encoder transmits to the relay via one use of a discrete memoryless channel (DMC) with transition law $P$. Let $X_i$ denote the input from the encoder and $Y_i$ denote the output received by the relay.
        \item The relay transmits to the decoder via one use of another DMC with the same transition law $P$. Let $W_i$ denote the input from the relay and $Z_i$ denote the output received by the decoder.
    \end{itemize}
    Importantly, $W_i$ must only be a function of $Y_1, \ldots, Y_{i-1}$; the relay is not allowed to use information from the future.  At time $n$, having received $Z_1, \ldots, Z_n$, the decoder forms an estimate of $\Theta$, which we denote by $\hat{\Theta}_n$ (or sometimes simply $\hat{\Theta}$). 
    
    The input alphabets and output alphabets of $P$ are denoted by $\calX_P$ and $\calY_P$ respectively (and similarly for other DMCs, e.g., $\calX_Q$, $\calY_Q$). We will also write these as $\calX$ and $\calY$ when there is no ambiguity.
    
    Let $P_e^{(2)}(n, M, P) = \p(\hat{\Theta}_n \ne \Theta)$ be the error probability with $n$ time steps and $M$ messages.  Then, the two-hop error exponent is defined as
    \begin{equation}
        \ratetwo_{M,P} = \sup_{{\rm protocols}} \liminf_{n\rightarrow \infty} \left\{ -\frac1n \log P_e^{(2)}(n, M, P)\right\},
        \label{learning_rate}
    \end{equation}
    where the supremum is over all possible designs of the encoder, relay, and decoder.
    
    We are also interested in the zero-rate error exponent. For $R>0$, we define the error exponent at rate $R$ by 
	\begin{equation}
	    \eratetwo_P(R) = \sup_{{\rm protocols}} \liminf_{n \rightarrow \infty} \left\{ -\frac1n \log P_e^{(2)}(n, e^{nR}, P) \right\},
	\end{equation}
	and we extend this function to $R=0$ via $\eratetwo_P(0) = \lim_{R \to 0^+} \eratetwo_P(R)$.  We sometimes omit the subscript $P$ and simply write $\ratetwo_{M}$, $\eratetwo$ and $P_e^{(2)}(n, M)$ when there is no ambiguity.
    
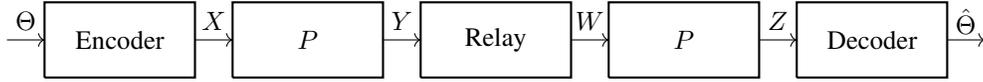
\begin{figure*}[!t]
    \centering
    \begin{tikzpicture}
\draw (0.75,1.25) node {$\Theta$};
\draw (3.25,1.25) node {$X$};
\draw (5.75,1.25) node {$Y$};
\draw (8.25,1.25) node {$W$};
\draw (10.75,1.25) node {$Z$};
\draw (13.25,1.25) node {$\hat{\Theta}$};
\draw[->] (0.5,1) -- (1,1);
\draw[thick] (1,0.5) -- (1,1.5) -- (3,1.5) -- (3,0.5) -- (1,0.5);
\node at (2, 1) {Encoder};
\draw[->] (3,1) -- (3.5,1);
\draw[thick] (3.5,0.5) -- (3.5,1.5) -- (5.5,1.5) -- (5.5,0.5) -- (3.5,0.5);
\node at (4.5, 1) {$P$};
\draw[->] (5.5,1) -- (6,1);
\draw[thick] (6,0.5) -- (6,1.5) -- (8,1.5) -- (8,0.5) -- (6,0.5);
\node at (7, 1) {Relay};
\draw[->] (8,1) -- (8.5,1);
\draw[thick] (8.5,0.5) -- (8.5,1.5) -- (10.5,1.5) -- (10.5,0.5) -- (8.5,0.5);
\node at (9.5, 1) {$P$};
\draw[->] (10.5,1) -- (11,1);
\draw[thick] (11,0.5) -- (11,1.5) -- (13,1.5) -- (13,0.5) -- (11,0.5);
\node at (12, 1) {Decoder};
\draw[->] (13,1) -- (13.5,1);
\end{tikzpicture}
    \caption{Illustration of our problem setup.}
    \label{fig:setup}
\end{figure*}
	
    Similar quantities can be defined for the one-hop case where encoder transmits directly to the decoder through $P$; we refer to the associated error exponents as $\rateone_M = \rateone_{M,P}$ and $\erateone = \erateone_{P}$.
    
    It is clear from data processing inequalities that $\rateone_{M,P} \geq \ratetwo_{M,P}$ and $\erateone_P(R) \geq \eratetwo_P(R)$. In this paper, we will derive various sufficient conditions under which the 1-hop and 2-hop error exponents are equal.
    
    \subsection{Related Work} \label{sec:existing}
    
    {\bf Point-to-point settings.} As summarized in \cite[Ch.~5]{gallager}, there are two particularly well-known achievable 1-hop error exponents at positive rates. The \textit{random coding exponent},  as its name suggests, is the error exponent of an optimal decoder when the codebook is generated in an i.i.d. manner. However, at low rates, the error probability is dominated by a small fraction of the codewords. Accordingly, improvements can be attained by generating an i.i.d. codebook and then removing the ``bad'' codewords; the error exponent formed by this procedure is called the \textit{expurgated exponent}.  Converse bounds on the positive-rate exponent (e.g., the sphere-packing bound) are also summarized in \cite[Ch.~5]{gallager}, but these are generally less relevant to our work, other than the fact that the expurgated exponent is tight for any DMC in the limit of zero rate.
    
    In \cite{berlekamp}, 1-hop error exponent bounds were derived for fixed $M$. The authors of \cite{berlekamp} introduced the notion of pairwise reversible channels (see Definition \ref{pairwise_reversible} below) and derived an achievable exponent which is tight for pairwise reversible channels. They also showed that as $M \rightarrow \infty$, the optimal error exponent approaches the zero-rate expurgated exponent, regardless of whether the pairwise reversible assumption holds.  We will use several of the results from \cite{berlekamp} as building blocks towards our own.
    
    {\bf Relay channel settings.} As mentioned in the introduction, in the 2-hop setting, the problem of relaying a single bit (i.e. $M=2$) over a binary symmetric channel (BSC) was studied in \cite{onebit,jog2020teaching,teachlearn}.  A variety of achievability bounds were developed in \cite{onebit,jog2020teaching} using various techniques (e.g., direct forwarding, relaying the best guess so far, and others) that we do not detail here; they have varying degrees of tightness, but all fall short of the simple converse based on the data processing inequality.  In \cite{teachlearn}, we showed that in fact the 1-hop and 2-hop error exponents match whenever $M=2$ (i.e. $\ratetwo_{2,P} = \rateone_{2,P}$), not only for the BSC but for any DMC $P$.  We also identified sufficient conditions under which the 1-hop and 2-hop exponents match when the two channels are different.
    
    While the study of error exponents for relay channels is already well-motivated in itself from a theoretical standpoint, we also note that Huleihel, Polyanskiy, and Shayevitz \cite{onebit} significantly strengthened the motivation by connecting the 1-bit 2-hop problem with the \emph{information velocity} problem, which was posed by Yury Polyanskiy \cite{onebit} and is also captured under a general framework studied by Rajagopalan and Schulman \cite{schulman_1994}.  Briefly, the goal is to reliably transmit a single bit over a long chain of relays while maintaining a non-vanishing ratio between the number of hops and the total transmission time.  Based on this connection, it was conjectured in \cite{onebit} that the 1-hop and 2-hop error exponents should coincide in the high-noise limit (so that ``information propagation does not slow down''), and as a step towards this conjecture, they showed that the two differ by at most a factor of $\frac{3}{4}$.  Our above result from \cite{teachlearn} confirmed their conjecture, without requiring the high-noise condition.  The results of the present paper may similarly have interesting connections with a multi-bit generalization of the information velocity problem, though we do not attempt to explore this direction here.

   Several other works have focused on error exponents for relay channels at positive rates, e.g., see \cite{highrates,bradford2012error,multihopping} and the references therein.  This is a fundamentally different regime from the constant-$M$ setting or zero-rate limit that we study; for example, \cite{highrates} uses random coding techniques, but as we mentioned above, it is well-known that random coding exponents are loose at low rates.

   The work \cite{multihopping} focuses on multi-hop tandem channels, and notes that concatenated codes \cite{concat} are optimal to within a factor of two in the zero-rate limit, while being much worse at higher rates.  Accordingly, they propose strategies with improved error exponents at these higher rates.  We emphasize that in our work, we are interested in scenarios where we can bring the factor of two all the way down to one, i.e., matching 1-hop and 2-hop exponents.  We are not aware of any results of this kind beyond the case of $M=2$ from \cite{teachlearn}.
   
   Additional works for other settings related to the preceding two paragraphs (e.g., Gaussian and/or fading channels) can be found in the reference lists of \cite{highrates,bradford2012error,multihopping,endotend}.  Finally, we briefly mention that other considerations involving relaying have included channel capacity \cite[Ch.~16]{gamalkim},  second-order asymptotics \cite{fong2017achievable}, and hypothesis testing \cite{hyptestrelay}.
        
    \subsection{Main Results}

    The result $\ratetwo_{2,P} = \rateone_{2,P}$ from \cite{teachlearn} naturally leads to the following question: {\em For what $(P,M)$ do we have $\ratetwo_{M,P} = \rateone_{M,P}$ when $M > 2$?}

    Towards partially answering this question, we first state the following definition.
    
    \begin{dfn}
    {\em \cite{berlekamp}}
    A discrete memoryless channel is \textbf{pairwise reversible} if, for all $x, x' \in \calX_P$, the quantity
    \begin{equation}
        \sum_{y \in \calY_{P}} P(y|x)^{1-s} P(y|x')^s
    \end{equation}
    attains its minimum at $s=\frac{1}{2}$ (possibly non-uniquely).
    \label{pairwise_reversible}
    \end{dfn}

    The class of pairwise reversible channels includes the BSC and BEC, as well as the `$K$-ary symmetric channel' where all diagonal entries take on a single value and all off-diagonal entries take on another value.  A key benefit of pairwise reversibility is that it leads to a straightforward calculation for $\rateone_{M,P}$, which is not available for general channels.  A more detailed discussion is given in \cite{berlekamp}, so we do not go into further detail here.
    
    Here we formally state the simplest forms of our results, while providing forward references to additional theorems that generalize these.  We first have the following
    
    \begin{theorem}
    For any pairwise reversible discrete memoryless channel $P$ and any fixed number of messages $M$, we have $\ratetwo_{M,P} = \rateone_{M,P}$.
    \label{thm:fixed_alphabet}
    \end{theorem}
	
    Theorem \ref{thm:fixed_alphabet} is proved in Section \ref{sec:protocol} after establishing some preliminary results in Section \ref{sec:prelim}.  The proof is based on a protocol that is explicit (i.e., constructive) other than using a codebook attaining the optimal exponent in the 1-hop setting as a black box.  Explicit constructions of such codebooks are known for pairwise reversible channels \cite[p.~431]{berlekamp}.  On the other hand, in Section \ref{sec:non_pw}, we give a generalization to a class of channels that need not be pairwise reversible (see Theorem \ref{thm:match}), and for such channels, explicit codebook constructions with optimal error exponents are generally unavailable.
 
 
    While one may hope based on these results (and those of \cite{teachlearn}) that $\ratetwo_{M,P} = \rateone_{M,P}$ for all $(M,P)$, the following result (proved in Section \ref{sec:differ}) shows that this is not the case. 

    \begin{theorem}
        In the case that $M=3$, there exist DMCs with $|\calX_P|=3$ and $|\calY_P|=4$ such that $\ratetwo_{3,P} < \rateone_{3,P}$.
        \label{thm:mismatch}
    \end{theorem}
	
    Next, we consider the \emph{zero-rate error exponent}, by which we mean the limit of the positive-rate exponent as the rate $R$ approaches zero from above. Unlike the fixed-$M$ case, the zero-rate error exponent for any DMC can easily be computed (see Theorem \ref{zero_rate}). The following theorem shows that in this regime, the 1-hop and 2-hop exponents are equal, without any assumptions on $P$.
	
    \begin{theorem}
        For any discrete memoryless channel $P$, we have $\eratetwo_P(0) = \erateone_P(0)$.
        \label{thm:zero_rate}
    \end{theorem}
	
    Theorem \ref{thm:zero_rate} is proved in Section \ref{zero_rate}, and uses Theorem \ref{thm:fixed_alphabet} as a stepping stone.  In contrast with Theorem \ref{thm:fixed_alphabet}, the proof of Theorem \ref{thm:zero_rate} is highly non-constructive; the reason for this is highlighted in the proof itself.
    
    Finally, while the preceding results focus on the case that the two channels in the system are identical, we present generalizations of Theorems \ref{thm:fixed_alphabet} and \ref{thm:zero_rate} in Section \ref{sec:diff_chan} that do not require this assumption.
	
\section{Preliminaries} \label{sec:prelim}

In this section, we introduce some additional notation and definitions, and provide a number of useful auxiliary results that will be used for proving our main results.

\subsection{Notation and Definitions}

To lighten notation, we let $P_x(\cdot)$ denote the output distribution $P(\cdot | x)$.  An \emph{$(M,\ell)$-codebook} is defined to be a collection of $M$ codewords each having length $\ell$, and when utilizing such a codebook, we will use the notation $(x^{(1)}, \ldots, x^{(M)})$ for the associated codewords.

For two probability distributions $Q, Q'$ over some finite set $\calX$, the Bhattacharyya distance is defined as
\begin{equation}
    \db(Q, Q') = -\log \sum_{x \in \calX} \sqrt{Q(x)Q'(x)}.
\end{equation}
For $x, x' \in \calX_P$, we also define the Bhattacharyya distance associated with two channel inputs as 
\begin{equation}
    \db(x, x', P) = \db(P_x, P_{x'}) \label{eq:dB_channel}
\end{equation}
with a slight abuse of notation.

Generalizing the Bhattacharyya distance, the Chernoff divergence with parameter $s$ is given by
\begin{equation}
    \dc(Q, Q',s) =  -\log \sum_{x \in \calX}  Q(x)^{1-s}Q'(x)^{s}, \label{eq:dc}
\end{equation}
and the Chernoff divergence (with optimized $s$) is given by
\begin{equation}
    \dc(Q, Q') = \max_{0\leq s \leq 1} \dc(Q, Q', s).
\end{equation}
Analogous to \eqref{eq:dB_channel}, we also write
\begin{equation}
    \dc(x, x', P) = \dc(P_x, P_{x'}). \label{eq:dC_channel}
\end{equation}
Note that whenever $P$ is pairwise reversible, we have $\dc(x, x', P) = \db(x, x', P)$.

For any positive integer $k$, we let $P^k$ denote the $k$-fold product of $P$, with probability mass function
\begin{equation}
    P^k(\vec{y}|\vec{x}) = \prod_{i=1}^k P(y_i|x_i).
\end{equation}
For two sequences $\vec{x}, \vec{x}'$ of length $k$, we also use the notation $\db(\vec{x}, \vec{x}', P^k)$ and $\dc(\vec{x}, \vec{x}', P^k)$ similarly to \eqref{eq:dB_channel} and \eqref{eq:dC_channel}, with the understanding that $\vec{x}, \vec{x}'$ are treated as inputs to $P^k$.

Next, for $S \subseteq \calX$, define
\begin{equation}
    \dbmin(S, P) = \min_{\substack{x,x' \in S, \\ x\neq x'}} \db(x, x', P),
\end{equation}
and similarly
\begin{equation}
    \dcmin(S, P) = \min_{\substack{x,x' \in S, \\ x\neq x'}} \dc(x, x', P).
    \label{chernoff_set}
\end{equation}
Given an $(M,k)$-codebook $\calC = (x^{(1)}, \ldots, x^{(M)})$, we similarly write 
\begin{align}
    \dbmin(\calC, P^k) &= \min_{m,m'} \db(x^{(m)}, x^{(m')}, P^k), \\
    \dcmin(\calC, P^k) &= \min_{m,m'} \dc(x^{(m)}, x^{(m')}, P^k),   
\end{align}
where the codewords are treated as inputs to $P^k$.

We let $\calP(\calX)$ denote the set of all probability distributions on $\calX$.  If $P_1$ and $P_2$ are channels such that $\calY_{P_1} \subseteq \calX_{P_2}$, define $P_2 \circ P_1$ to be the composite channel formed by feeding the output of $P_1$ into $P_2$.  We will also treat deterministic functions as channels and write $f\circ P_1$ where appropriate.
    
\subsection{Auxiliary Results}
\subsubsection{Results on 1-hop error exponents}
We will use two results from \cite{berlekamp} regarding the 1-hop error exponents.

\begin{theorem}
    \emph{\cite[Thm.~2]{berlekamp}}
    For any $\edag < \rateone_M$, it holds for all sufficiently large $\ell$ that there exists an $(M,\ell)$-codebook $\calC$ such that
    \begin{equation}
        \dcmin(\calC, P^\ell) \geq \ell \cdot \edag.
        \label{eq:berlekamp}
    \end{equation}
    \label{thm:berlekamp}
\end{theorem}

\begin{theorem}
    \emph{\cite[Thm.~4]{berlekamp}} For any $P$, the zero-rate error exponent is given by
    \begin{equation}
        E^{(1)}(0) = \max_{q \in \calP(\calX)} \sum_{x, x' \in \calX} q_{x} q_{x'} \db(x, x', P). \label{eq:E1zero}
    \end{equation}
    \label{zero_rate}
\end{theorem}

\subsubsection{Properties of Chernoff divergence}
The following lemma concerns the Chernoff divergence of a composite channel; we are not aware of a reference for this result, so we provide a complete proof.

\begin{lemma}
    Let $P_1, P_2$ be channels such that $\calY_{P_1} \subseteq \calX_{P_2}$. For all $x, x' \in \calX_{P_1}$, we have
    \begin{multline}
        \dc(x, x', P_2 \circ P_1, s) \geq \min_{y, y' \in \calY_{P_1}} \big\{ \dc(y, y', P_2, s) \\ -(1-s) \log P_1(y|x) - s \log P_1(y'|x') \big\} - 2 \log |\calY_{P_1}|.
    \end{multline}
    \label{lem:distributive}
\end{lemma}
\begin{proof}
   Since $(a+b)^s \leq a^s + b^s$ for $s \in [0,1]$, we have for all $z \in \calY_{P_2}$ that
    \begin{align}
    (P_2 \circ P_1)(z|x)^s 
        &= \left(\sum_{y \in \calY_{P_1}} P_2(z|y) P_1(y|x)\right)^s \\
        &\leq \sum_{y \in \calY_{P_1}} P_2(z|y)^s P_1(y|x)^s,
    \end{align}
    and similarly
    \begin{equation}
    (P_2 \circ P_1)(z|x')^{1-s} \leq \sum_{y' \in \calY_{P_1}} P_2(z|y')^{1-s} P_1(y'|x')^{1-s}.
    \end{equation}
    It follows that
    \begin{align}
        &\sum_{z \in \calY_{P_2}} (P_2 \circ P_1)(z|x)^{1-s} (P_2 \circ P_1)(z|x')^{s}\\
        &\leq \sum_{z \in \calY_{P_2}} \bigg(\sum_{y \in \calY_{P_1}} P_2(z|y)^{1-s} P_1(y|x)^{1-s}  \nonumber \\ 
            &\hspace*{3cm} \times\sum_{y' \in\calY_{P_1}} P_2(z|y')^{s} P_1(y'|x')^{s} \bigg)\\
        &= \sum_{y, y' \in \calY_{P_1}} P_1(y|x)^{1-s} P_1(y'|x')^{s} \sum_{z \in \calY_{P_2}} P_2(z|y)^{1-s} P_2(z|y')^{s} \\
        &\leq |\calY_{P_1}|^2 \max_{y, y' \in \calY_{P_1}} P_1(y|x)^{1-s} P_1(y'|x')^{s} \nonumber \\ 
            &\hspace*{3cm} \times \sum_{z \in \calY_{P_2}} P_2(z|y)^{1-s} P_2(z|y')^{s}.
    \end{align}
    Taking the negative logarithm on both sides gives the desired result.
\end{proof}

Next, we state a simple tensorization property of Chernoff divergence.  This result is standard, but we provide a short proof in Appendix \ref{sec:pf_tensorize}.

\begin{lemma} \label{lem:tensorize}
    For any sequences $\vec{x} = (x_1,\dotsc,x_k)$ and $\vec{x}' = (x'_1,\dotsc,x'_k)$, we have
    \begin{equation}
        \dc(\vec{x}, \vec{x}', P^k, s) = \sum_{i=1}^k \dc(x_i, x'_i, P, s),
        \label{eq:iid1}
    \end{equation}
    and
    \begin{equation}
        \dc(\vec{x}, \vec{x}', P^k) = \max_{0\leq s\leq 1} \sum_{i=1}^k \dc(x_i, x'_i, P, s).
        \label{eq:iid2}
    \end{equation}
    \label{lem:iid}
\end{lemma}

Note that in general, it can happen that $\dc(\vec{x}, \vec{x}', P^k) \ne  \sum_{i=1}^k \dc(\vec{x}_i, \vec{x}'_i, P)$, due to the order of summation and maximization.

The following lemma gives a useful relation between the KL divergence and Chernoff divergence.  We are again not aware of an existing statement of this result, so we provide a short proof.

\begin{lemma}
    For any three distributions $Q_1, Q_2, Q_3$ defined over the same finite alphabet and any $0\leq s \leq 1$, we have
    \begin{equation}
        (1-s) D(Q_1\|Q_2) + sD(Q_1\|Q_3) \geq \dc(Q_2,Q_3, s).
    \end{equation}
    \label{minimum_value}
\end{lemma}
\begin{proof}
    We have
    \begin{align}
        & (1-s) D(Q_1\|Q_2) + sD(Q_1\|Q_3) \nonumber \\
        &= (1-s) \sum_x Q_1(x) \log \frac{Q_1(x)}{Q_2(x)} +s \sum_x Q_1(x) \log \frac{Q_1(x)}{Q_3(x)}\\
        &= -\sum_x Q_1(x) \log \frac{{Q_2(x)^{1-s}Q_3(x)^s}}{Q_1(x)} \\ 
        &\geq -\log\Big(\sum_x Q_2(x)^{1-s}Q_3(x)^s\Big) = \dc(Q_2, Q_3,s),
    \end{align}
    where we applied Jensen's inequality on the convex function $-\log(\cdot)$.
\end{proof}

\subsubsection{Divergence-based bounds on probabilities}

The following result bounds the error exponent of a DMC in terms of the Chernoff divergence.  This result is implicit in prior works such as \cite{berlekamp}, but we also provide a short proof in Appendix \ref{sec:pf_initial}.

\begin{lemma} \label{lem:initial}
    Given $P$ and $M$, for any $S \subseteq \calX_P$ with $|S|=M$, we have
    \begin{equation}
        \rateone_{M,P} \geq \dcmin(S,P).
    \end{equation}
    Moreover, this lower bound can be attained using repetition coding, in which the message is encoded by repeating a corresponding element of $S$.
\end{lemma}

We will also use a well-known result on the probability of falling within a given type class.

\begin{lemma}
    \emph{(\cite[Theorem 11.1.4]{cover_thomas})} Let $Y_1, \ldots, Y_n$ be i.i.d.~random variables with distribution $P_Y$, and let $Q$ be the empirical distribution of $(Y_1,\ldots, Y_n)$. Then, for any $q \in \calP(\calY)$, we have
    \begin{equation}
        -\log \p(Q=q) \geq n\cdot D(q\|P_Y).
    \end{equation}
    \label{lem:type_class}
\end{lemma}

\section{Protocol Design and Analysis (Proof of Theorem \ref{thm:fixed_alphabet})} \label{sec:protocol}

In this section, we introduce our protocol for the fixed-$M$ regime, first at a high level and then with specific details, leading to a proof of Theorem \ref{thm:fixed_alphabet}.

\subsection{Block-Structured Protocol: Macroscopic View}

The high-level view of our protocol is similar to the case of $M=2$ \cite{teachlearn}, but the details are largely different with several new challenges; see Section \ref{sec:relation} for some discussion and comparison.

Let $f: \calY^k \rightarrow \calX^k$ be a function on length-$k$ sequences, and let $\calC = (x^{(1)}, \ldots, x^{(M)})$ be an arbitrary $(M,k)$ codebook. Consider the following block structured protocol:
\begin{itemize}
    \item Upon receiving $\Theta \in \{1,\ldots, M\}$, the encoder repeatedly sends $x^{(\Theta)}$ in blocks of $k$ symbols.
    \item The relay reads in blocks of length $k$ and sends to the decoder in blocks of $k$, using the mapping
    \begin{equation}
        W_{[ik+1, (i+1)k]} = f(Y_{[(i-1)k+1, ik]})
    \end{equation}
    for all $i\geq 1$. The first block $W_{[1,k]}$ is arbitrary, and is ignored by the decoder.
    \item After receiving $Z_{[1,n]}$, the decoder forms the estimate $\hat{\Theta}_n$; we will focus on maximum-likelihood decoding.
\end{itemize}

We momentarily ignore the fact that the encoder uses block-wise repetition coding, and consider the effect of the relay's strategy.  We see that its strategy leads to the encoder being able to send ``directly'' to the decoder using $\lfloor n/k-1 \rfloor$ uses of the composite channel $P^k \circ f \circ P^k$.  Hence, by considering an optimal encoder/decoder pair, we deduce that
\begin{equation}
    P^{(2)}_e(n, M, P) \leq P^{(1)}_e(\lfloor n/k-1 \rfloor, M, P^k\circ f \circ P^k).
    \label{eq:two_to_one}
\end{equation}
Accordingly, we can bound the 2-hop error exponent in terms of a 1-hop error exponent:
\begin{align}
    \ratetwo_{M,P} &= \liminf_{n\rightarrow \infty}  -\frac1n \log P^{(2)}_e(n, M, P) \\
    &\geq \liminf_{n\rightarrow \infty} -\frac1n \log P^{(1)}_e(\lfloor n/k-1 \rfloor, M, P^k\circ f \circ P^k)\\
    &= \liminf_{n\rightarrow \infty} -\frac{\lfloor n/k-1 \rfloor}n \frac{1}{\lfloor n/k-1 \rfloor} \nonumber \\
        &\hspace*{1cm}\times\log P^{(1)}_e(\lfloor n/k-1 \rfloor, M, P^k\circ f \circ P^k)\\
    & = \frac1k \rateone_{M,P^k \circ f \circ P^k}. \label{eq:E2_to_E1}
\end{align}
Now consider the $(M,k)$-codebook $\calC$ introduced above. Since the codewords of $\calC$ are inputs to $P^k$ (and hence to $P^k \circ f \circ P^k$), we can apply Lemma \ref{lem:initial} to obtain
\begin{equation}
    \rateone_{M, P^k \circ f \circ P^k} \geq \dcmin(\calC, P^k \circ f \circ P^k). \label{eq:E2_to_E2_weakened}
\end{equation}
Next, we discuss the role of the encoder.  In principle, to attain the lower bound in \eqref{eq:E2_to_E1}, the encoder (and decoder) may need to perform complicated coding over the ``super-alphabet'' $\calX^k$.  However, since Lemma \ref{lem:initial} is based on repetition coding, such coding is no longer necessary for the weakened lower bound \eqref{eq:E2_to_E2_weakened}, and instead our protocol described above suffices.  Thus, we have established the following.

\begin{lemma}
    Given $M$ and $P$, we have for any $(k,\calC,f)$ that
    \begin{equation}
        \ratetwo_{M,P} \geq 
        \frac1k \dcmin(\calC, P^k \circ f \circ P^k).
    \end{equation}
    \label{lem:block}
    Moreover, this lower bound is achieved by our protocol described above.
\end{lemma}

\subsection{Details and Analysis of the Protocol when \texorpdfstring{$|\calX| = M$}{X=M}}
\label{subsection:protocol_details}

 Throughout this subsection, we adopt the additional assumption $|\calX| = M$; given this assumption, we may set $\calX = \{1,\dotsc,M\}$ without loss of generality.  This turns out to be a convenient stepping stone towards proving the general case in the following subsection. Although Theorem \ref{thm:fixed_alphabet} assumes that the channels are pairwise reversible, all results in this subsection do not require this assumption.

Define the shorthand
\begin{equation}
    \altrate = \min_{m\neq m'} \db(m,m',P). \label{eq:shorthand}
\end{equation}
The function $f$ in our protocol maps sequences to sequences, but will be defined via another function $g$ mapping distributions to distributions.  Specifically, letting $m_q = \argmin_j D(q\|P_j)$ for each $q \in \calP(\calY)$ (with arbitrary tie-breaking), we define $g: \calP(\calY) \rightarrow \calP(\calX)$ as follows:

\begin{equation}
    g(q)_x = \begin{cases}
    \frac1{|\calX|} \min\Big(\frac{D(q\|P_{m_q})}{\altrate}, 1\Big) & x\neq m_q \\
    1-\frac{|\calX|-1}{|\calX|} \min\Big(\frac{D(q\|P_{m_q})}{\altrate}, 1\Big) & x = m_q.
    \end{cases}
    \label{eq:g_dfn}
\end{equation}
In the case of a binary symmetric channel with $M=2$, this function is closely related (but not identical) to that used in our earlier work \cite{teachlearn}; see Section \ref{sec:relation} for further discussion. 

We first establish a useful property of $g$.
\begin{lemma}
    With $|\calX| = M$, for all $q \in \calP(\calX)$, $m, m' \in \calX$ and $m\neq m'$, we have
    \begin{equation}
        D(q\|P_m) \geq \altrate (1-g(q)_m + g(q)_{m'}).
        \label{eq:linear_bound}
    \end{equation}
    \label{lem:linear_bound}
\end{lemma}
\begin{proof}
    (\underline{Case 1}: $m_q=m$). From the definition of $g$ in \eqref{eq:g_dfn}, we obtain
    \begin{equation}
        \altrate (1-g(q)_m + g(q)_{m'}) = \altrate \cdot \min\Big(\frac{D(q\|P_{m})}{\altrate}, 1\Big) \leq D(q\|P_m).
    \end{equation}
    (\underline{Case 2}: $m_q=m'$).
    By Lemma \ref{minimum_value} with $s = \frac{1}{2}$, we have
    \begin{equation}
        \frac12 D(q\|P_m) + \frac12 D(q\|P_{m'}) \geq \db(m, m', P) \geq \altrate.
        \label{eq:lem4case1}
    \end{equation}
    Moreover, since $m' = \argmin_j D(q\|P_j)$, we have
    \begin{equation}
        D(q\|P_m) \geq \frac12 D(q\|P_m) + \frac12 D(q\|P_{m'}) \geq \altrate.
        \label{eq:lem4case2}
    \end{equation}
    Therefore, \eqref{eq:g_dfn} gives
    \begin{align}
        \altrate (1-g(q)_m + g(q)_{m'}) 
            &= \altrate \cdot \left(2- \min\Big(\frac{D(q\|P_{m})}{\altrate}, 1\Big)\right) \\
            &= \max(\altrate, 2\altrate - D(q\|P_m)) \\ 
            &\leq D(q\|P_m),
    \end{align}
    where in the last step we apply \eqref{eq:lem4case2}.
    
    (\underline{Case 3}: $m_q \notin \{m, m'\}$). In this case, we have $D(q\|P_m) \geq \altrate$ (using the same reasoning as \eqref{eq:lem4case2}) and $g(q)_m=g(q)_{m'}$, and the conclusion follows immediately.
\end{proof}

Next, to each $q \in \calP(\calY)$, we associate a codeword $w^{(q)}$.  Up to rounding issues, this codeword simply repeats each $x \in \calX$ for $k \cdot g(q)_x$ times, and thus takes a form such as $1112222222233$.  To account for rounding, the general procedure for constructing $w^{(q)}$ is as follows: {\em For each $i = 1,\ldots, M$, append 
\begin{equation}
    \left\lfloor k \sum_{i'\leq i} g(q)_{i'} \right\rfloor - \left\lfloor k \sum_{i' < i} g(q)_{i'} \right\rfloor
\end{equation}
copies of symbol $i$, in sequence.}  From this construction, we have
\begin{equation}
    |\text{(number of times $x$ appears in $w^{(q)}$)} - k\cdot g(q)_x| \leq 1.
    \label{eq:rounding_error}
\end{equation}
Observe that $w \,:\, \calP(\calY) \to \calX^k$ maps a probability distribution on $\calY$ to a length-$k$ sequence on $\calX$.  We also introduce the function $\empdist \,:\, \calY^k \to \calP(\calY)$ that simply maps a length-$k$ sequence to its type (i.e., its empirical distribution).  With these definitions in place, we can specify our choice of $f$ as follows:
\begin{equation}
    f(\vec{y}) = w^{(\hat{q})}, \quad \text{where}~~ \hat{q} = \empdist(\vec{y}). \label{eq:choice_f}
\end{equation}
That is, $f$ first maps the received length-$k$ sequence to an empirical distribution, which in turn is mapped to the associated $w^{(\cdot)}$ sequence.

Recalling that $\calX = \{1,\dotsc,M\}$ in this subsection, we consider the `trivial' $(M,k)$-codebook $\calC = (x^{(1)}, \ldots, x^{(M)})$ where $x^{(m)}$ is simply $k$ copies of symbol $m$.  Then, we have the following.

\begin{lemma}
    With $\calX = \{1,\dotsc,M\}$, let $f$ be defined as above and $\calC$ be the `trivial' $(M,k)$-codebook. Then, we have
    \begin{equation}
        \dbmin(\calC, P^k \circ f \circ P^k) \geq (k-2) \cdot \altrate - 2|\calY| \log(k+1).
    \end{equation}
    \label{lem:db_sep}
\end{lemma}
\begin{proof}
Throughout the proof, we let $\calP_k(\calY)$ denote the set of all empirical distributions (i.e., types) associated with sequences in $\calY^k$.  
Observe that our choice of $f$ in \eqref{eq:choice_f} gives $P^k \circ f \circ P^k = (P^k \circ w) \circ (\empdist \circ P^k)$.  As a result, we can apply Lemma \ref{lem:distributive} to obtain
\begin{align}
    &\db(x^{(m)}, x^{(m')}, P^k \circ f \circ P^k) \geq -2 |\calY| \log (k+1) \nonumber \\
    & + \max_{q,q' \in \calP_k(\calY)} \bigg\{ \db(q, q', P^k \circ w) - \frac12 \log\big[(\empdist\circ P^k)(q|x^{(m)})\big] \nonumber \\
    &\hspace*{3cm}- \frac12 \log\big[(\empdist\circ P^k)(q'|x^{(m')})\big] \bigg\},
    \label{eq:channel_decompose}
\end{align}
where we note the following:
\begin{itemize}
    \item $\empdist \circ P^k$ and $P^k \circ w$ play the role of $P_1$ and $P_2$ in Lemma \ref{lem:distributive}, and the channels are unconventional in the sense that the output alphabet of $P_1$ and input alphabet of $P_2$ are both $\calP_k(\calY)$.
    \item The first term in \eqref{eq:channel_decompose} arises because $|\calP_k(\calY)| \le (k+1)^{|\calY|}$ \cite[Ch.~2]{ck}.
    \item $x^{(m)}, x^{(m')}$ replaces $(x,x')$ in Lemma \ref{lem:distributive}, and $q,q'$ replaces $y,y'$.
\end{itemize}
We now bound the terms appearing in \eqref{eq:channel_decompose}.  By Lemma \ref{lem:type_class}, and recalling that $P_m(\cdot) = P(\cdot|m)$ (with $\calX = \{1,\dotsc,M\}$ in this subsection) and the use of a trivial codebook, we have
\begin{equation}
    -\log\big[(\empdist\circ P^k)(q|x^{(m)})\big] \geq k\cdot D(q\|P_m).
    \label{eq:sanov1}
\end{equation}
and similarly,
\begin{equation}
     -\log [(\empdist\circ P^k)(q'|x^{(m')})] \geq k\cdot D(q'\|P_{m'}).
     \label{eq:sanov2}
\end{equation}
To bound $\db(q, q', P^k \circ w)$ in \eqref{eq:channel_decompose}, we write
\begin{align}
    \db(q, q', P^k \circ w) 
    &= \db(P^k_{w^{(q)}}, P^k_{w^{(q')}}) \label{eq:remove_w} \\ 
    &= \sum_{i=1}^k \db(P_{w^{(q)}_i}, P_{w^{(q')}_i}) \label{eq:tensorize} \\
    &\geq \sum_{i=1}^k \boldsymbol{1}(w_i^{(q)} \ne w_i^{(q')}) \cdot \altrate \label{eq:expand_indicator} \\
    &\geq (k\cdot \|g(q)-g(q')\|_\infty - 2) \altrate, \label{eq:sup_norm}
\end{align}
where:
\begin{itemize}
    \item \eqref{eq:remove_w} follows since $w$ is a deterministic function.
    \item \eqref{eq:tensorize} follows from the tensorization of $\db$ (see Lemma \ref{lem:tensorize}).
    \item \eqref{eq:expand_indicator} follows since the Bhattacharyya distance is non-negative, and is at least $\altrate$ when the inputs differ (see \eqref{eq:shorthand}).

    \item \eqref{eq:sup_norm} holds because for all $x \in \calX$, the number of times $x$ appears in exactly one of $w^{(q)}$ and $w^{(q')}$ is at least the difference in the number of occurrences of $x$; the $-2$ term comes from the `rounding error' in  \eqref{eq:rounding_error}.
\end{itemize}

Substituting \eqref{eq:sanov1}, \eqref{eq:sanov2}, and \eqref{eq:sup_norm} into \eqref{eq:channel_decompose}, we obtain
\begin{align}
   & \db(x^{(m)}, x^{(m')}, P^k) 
   \geq -2 |\calY|\log(k+1) -2 \altrate \nonumber \\ 
   &+ k \cdot \max_{q,q' \in \calP(\calY)}\Big\{ \|g(q)-g(q')\|_{\infty}\altrate + \frac12 D(q\|P_m) \nonumber \\
   &\hspace*{4.6cm}+ \frac12 D(q'\|P_{m'})\Big\}.
   \label{eq:db_bound}
\end{align}
Next, observe that
\begin{align}
    \|&g(q)-g(q')\|_{\infty} 
        = \max_{j} |g(q)_j - g(q')_j| \\
        & \ge \frac12(g(q)_m - g(q')_m) + \frac12(g(q')_{m'} - g(q)_{m'}),
\end{align}
and combining this with Lemma \ref{lem:linear_bound} gives
\begin{align}
    &\|g(q)-g(q')\|_{\infty}\altrate + \frac12 D(q\|P_m) + \frac12 D(q'\|P_{m'}) \\
    &\geq \altrate \Big(\frac12(g(q)_m - g(q')_m) + \frac12(g(q')_{m'} - g(q)_{m'})\Big) \nonumber \\
        &~~~ + \frac12 \altrate\big(2- g(q)_m + g(q)_{m'} -g(q')_{m'} + g(q')_{m}\big)\\
    &= \altrate
    \label{eq:overall_bound}
\end{align}
Combining \eqref{eq:db_bound} and \eqref{eq:overall_bound} completes the proof of Lemma \ref{lem:db_sep}.
\end{proof}

Combining Lemma \ref{lem:block} and Lemma \ref{lem:db_sep}, recalling the definition of $\altrate$ in \eqref{eq:shorthand}, and noting that $k$ can be arbitrarily large, we obtain the following.

\begin{theorem}
    In the case that $|\calX| = M$, we have
    \begin{equation}
        \ratetwo_{|\calX|} \geq \min_{m\neq m'} \db(m,m',P).
    \end{equation}
    \label{thm:m_inputs}
\end{theorem}

\subsection{Proof of Theorem \ref{thm:fixed_alphabet}} \label{sec:pf_fixed}

We are now ready to  remove the assumption that $\calX = M$, and prove Theorem \ref{thm:fixed_alphabet}.  
For any $\edag < \rateone_M$, fix an $(M,\ell)$-codebook $\calC'$ (with sufficiently large $\ell$) satisfying the conditions of Theorem \ref{thm:berlekamp}. Since $P$ is pairwise reversible, the quantity
\begin{equation}
    \dc(x^{(m)}, x^{(m')}, P^\ell) = \max_{0\leq s \leq 1} \sum_{i=1}^\ell \dc(x^{(m)}_i, x^{(m')}_i, P, s)
\end{equation}
(see Lemma \ref{lem:tensorize}) is maximized at $s=\frac12$.  Therefore, for all $m \ne m'$, we have
\begin{equation}
    \dc(x^{(m)}, x^{(m')}, P^\ell) = \db(x^{(m)}, x^{(m')}, P^\ell) \geq \ell \cdot \edag,
    \label{eq:strong_codebook}
\end{equation}
where the last inequality is guaranteed by the preceding application of Theorem \ref{thm:berlekamp}.

Let $P'$ be the restriction of $P^\ell$ to the codewords in $\calC'$, so that $\calX_{P'} = \calC'$ and $\calY_{P'} = \calY^{\ell}$.  Noting that $\ell$ is a fixed (albeit possibly large) constant, $P'$ can be viewed as a discrete memoryless channel (on a large alphabet).  In particular, $P'$ satisfies the condition for Theorem \ref{thm:m_inputs} with $|\calX_{P'}| = M$, as well as $\min_{m\neq m'} \db(x^{(m)}, x^{(m')},P') \geq \ell \cdot \edag$.  Hence, we have
\begin{equation}
    \ratetwo_{M, P'} \geq \ell  \cdot \edag. \label{eq:P'_lb}
\end{equation}
Now the idea is that we can run the protocol from the previous subsection with ``input alphabet'' $\calC'$ and ``output alphabet'' $\calY^{\ell}$, meaning that blocks of $\ell$ symbols are being treated as one ``super-symbol''. By doing so, we obtain
\begin{equation}
    P^{(2)}_e(n,M,P) \leq P^{(2)}_e(\lfloor n/\ell \rfloor,M,P^\ell) \leq P^{(2)}_e(\lfloor n/\ell \rfloor,M,P'), \label{eq:super_symbol}
\end{equation}
which implies via \eqref{eq:P'_lb} that
\begin{equation}
    \ratetwo_{M,P} \geq \frac1\ell\ratetwo_{M, P'} \geq \edag,
\end{equation}
as desired.

\subsection{Relation to Previous Work} \label{sec:relation}

The function $g(q)$ in \eqref{eq:g_dfn} is central to our analysis, and an analogous function played a similar role in our previous work on the BSC with $M=2$ \cite{teachlearn}.  

In fact, the latter turns out to be closely related to the former.  For the BSC with crossover probability $p$, the distributions $P_m$ are Bernoulli with parameters $p$ and $1-p$, and a simple calculation shows that $\altrate$ in \eqref{eq:shorthand} is equivalent to $D\big(\frac{1}{2} \| p\big)$ (where $D(a\|b)$ represents the KL divergence between Bernoulli distributions).  Hence, when $q$ corresponds to a fraction $(\alpha,1-\alpha)$ of the two symbols with $\alpha \le \frac{1}{2}$, the choice of $g(q)$ in \eqref{eq:g_dfn} corresponds to having $\frac{1}{2}\min\big(\frac{D(\alpha\|p)}{2D(1/2\|p)}, 1\big)$ of one symbol, and $1-\frac{1}{2}\min\big(\frac{D(\alpha\|p)}{2D(1/2\|p)}, 1\big)$ of the other.

In Figure \ref{fig:bsc}, we compare this choice of $g$ to the one made in our previous work \cite{teachlearn} for the BSC.  The latter is perhaps more intuitive due to being monotone -- if more 1s are received, then more 1s are sent.  This property is also likely to be beneficial in practical scenarios, since it leads to the relay indicating higher certainty in blocks that were lucky enough to have very few bit flips.  Despite this, it turns out that the two choices of $g$ have identical error exponents.  Essentially, using the monotone curve only reduces the probability associated with error events that are non-dominant in dictating the overall error exponent.

The use of a monotone curve also has a natural interpretation that the belief on $\Theta$ increases together with a suitably-defined likelihood ratio; we used this interpretation to provide a generalization beyond BSCs in \cite{teachlearn}.  The difficulty when $M > 2$ is that there are ${\binom{M}{2}}$ pairwise likelihood ratios of interest.  Accordingly, we were unable to find a ``neat'' solution that maintains a similar kind of monotonicity to the $M=2$ case; if we partially define a function on the simplex by starting with the lines between pairs of distributions $\{P_m\}_{m=1}^M$ (supposing $|\calX| = M$), we are left with considerable ``gaps'' in the simplex that are unclear how to fill.  Thus, we adopted the different approach of measuring confidence via the minimum KL divergence to the empirical distribution.  Lemma \ref{minimum_value} then helps to lower bound  the KL divergence with respect to the other inputs. 

    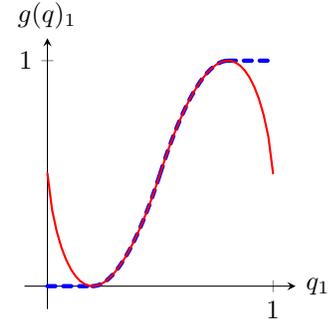
\begin{figure}
    \begin{center}
        \begin{tikzpicture}[line cap=round,line join=round,>=triangle 45,x=3.0cm,y=3.0cm]
            \draw (1.3,0.1) node{$q_1$};
            \draw (0.1,1.3) node{$g(q)_1$};

            \begin{axis}[
                x=3.0cm,y=3.0cm,
                axis lines=middle,
                xmin=-0.1,
                xmax=1.1,
                ymin=-0.1,
                ymax=1.1,
                xtick={0,1},
                ytick={0,1},]
                \clip(-0.1,-0.1) rectangle (1.1,1.1);
                \addplot[ultra thick,domain=0:0.2, color=blue, dashed] {0};
                \addplot[ultra thick,domain=0.2:0.5, color=blue, dashed] {(x*log10(x/0.2)+(1-x)*log10((1-x)/0.8))/0.1938};
                \addplot[ultra thick,domain=0.5:0.8, color=blue,dashed] {1-((1-x)*log10((1-x)/0.2)+x*log10(x/0.8))/0.1938};
                \addplot[ultra thick,domain=0.8:1, color=blue,dashed] {1};
                \addplot[thick, domain=0:0.5, color=red] {(x*log10(x/0.2)+(1-x)*log10((1-x)/0.8))/0.1938};
                \addplot[thick, domain=0.5:1, color=red] {1-((1-x)*log10((1-x)/0.2)+x*log10(x/0.8))/0.1938};
            \end{axis}
        \end{tikzpicture}
        \caption{The function $g$ used in the protocol, plotted for a BSC with crossover probability $p=0.2$. The dashed line corresponds to the strategy from our previous previous work \cite{teachlearn}, while the solid line represents the choice of $g$ in the present paper.}
        \label{fig:bsc}
    \end{center}
\end{figure}

\section{Channels that are not Pairwise Reversible} \label{sec:non_pw}

In this section, we show that the 1-hop and 2-hop exponents coincide for a broader class of channels that need not be pairwise reversible.  The main technical effort towards doing so is in deriving the following result for the case that $|\calX| = M$.

\begin{theorem} \label{thm:almost}
    Suppose that $|\calX| = M$ and $|\calX|\geq 3$, and let $t$ be the unique root in $\big[\frac13, \frac12\big]$ satisfying\footnote{The existence and uniqueness follows from the fact that the left-hand side of \eqref{inequality_t} is continuous and strictly increasing, and is below (respectively, above) one at $t = \frac{1}{3}$ (respectively, $t = \frac{1}{2}$).}
    \begin{equation}
        t + \frac{1}{M} + \frac{M-2}{M} \frac{t}{1-t} = 1.
        \label{inequality_t}
    \end{equation}
    Then, we have $\ratetwo_{|\calX|} \geq \altrate$, where
    \begin{equation}
        \altrate = \min_{m \neq m'} \max_{t \leq s \leq 1-t} \dc(m, m', P, s).
        \label{eq:altrate}
    \end{equation}
    \label{almost_reversible}
\end{theorem}
\begin{proof}
    See Section \ref{sec:pf_almost}.
\end{proof}

One way to interpret this theorem is to contrast it against Theorem \ref{thm:m_inputs}. There, we defined $\altrate$ with respect to $\db$, corresponding to $s=\frac{1}{2}$.  Here, we are allowed to take any value of $s$ between $t$ and $1-t$.

The assumption $|\calX| \ge 3$ is not restrictive, because our previous work established the exact exponent (and the fact that it matches the 1-hop setting) for $M=|\calX|=2$ and all $P$ \cite{teachlearn}.  As $M$ increases, $t$ decreases; the (approximate) threshold values for $t$ for some values of $M$ are given in Table \ref{tbl:t_vals}.

\begin{table}[!t]
\centering
\caption{Values of $t$ as a function of $M$. \label{tbl:t_vals}}
\begin{tabular}{c|c}
$M$   & $t$      \\ \hline
3      & 0.423  \\ \hline
4      & 0.407  \\ \hline
5      & 0.4    \\ \hline
$\to \infty$ & $\to$ 0.3819
\end{tabular} 
\end{table} \par

With Theorem \ref{almost_reversible} in place, we readily obtain the following in a similar manner to Section \ref{sec:pf_fixed} (wherein the proof of Theorem \ref{thm:fixed_alphabet} was completed).

\begin{theorem} \label{thm:match}
    Consider any DMC $P$ and any fixed number of messages $M$.  
    Let $t$ be defined via \eqref{inequality_t}, and suppose that for all $x \neq x'$, it holds that
    \begin{equation}
        t \leq \argmax_{0\leq s\leq 1} \dc(x, x', P, s) \leq 1-t.
        \label{condition_argmin}
    \end{equation}
    Then, we have $\ratetwo_M = \rateone_M$.
    \label{thm:non_reversible_equal_rate}
\end{theorem}
\begin{proof}
    See Section \ref{sec:pf_nonpw_rate}.
\end{proof}


\subsection{Proof of Theorem \ref{almost_reversible}} \label{sec:pf_almost}

We adopt the same protocol as in Section \ref{subsection:protocol_details} (using the same function $f$ and the `trivial' codebook $\calC$).  The definition of $g$ in \eqref{eq:g_dfn} is also the same, except that the more general definition of $\altrate$ in \eqref{eq:altrate} is used (the previous choice in \eqref{eq:shorthand} corresponds to setting $t = \frac{1}{2}$).

For each pair $(m,m')$ with $m \ne m'$, choose $\smm \in [t,1-t]$ such that $\dc(m, m', P, \smm) \geq \altrate$. We may assume that $\smm + s_{m',m} = 1$, since $\dc(m,m',s) = \dc(m',m,1-s)$ by the definition of $\dc$.

We will prove the following for all $m \ne m'$:
\begin{equation}
    \dc(x^{(m)}, x^{(m')}, P^k \circ f \circ P^k, \smm) \geq (k-2) \altrate - 2 |\calY| \log(k+1)
    \label{eq:dc_target}.
\end{equation}
Assuming \eqref{eq:dc_target} holds, we immediately obtain
\begin{equation}
    \dcmin(\calC, P^k \circ f \circ P^k) \geq (k-2) \altrate - 2 |\calY|\log(k+1),
\end{equation}
and applying Lemma \ref{lem:block} and taking $k \to \infty$ yields the desired bound $\ratetwo_M \geq \altrate$.

Recall from Section \ref{subsection:protocol_details} that $w^{(q)} \in \calX^k$ maps a distribution $q$ to an ordered sequence (using $g$ defined in \eqref{eq:g_dfn}), and that the function $\empdist(\cdot)$ maps a length-$k$ sequence to its empirical distribution.  As before, we choose $f = w \circ \hat{p}$. Towards establishing \eqref{eq:dc_target}, we first apply Lemma \ref{lem:distributive} to obtain the following generalization of \eqref{eq:channel_decompose}:
\begin{align}
    &\dc(x^{(m)}, x^{(m')}, P^k \circ f \circ P^k, \smm)  \nonumber \\ &\geq  -2 |\calY|\log (k+1) +  \min_{q,q' \in \calP(\calY)} \Big(\dc(q, q', P^k \circ w, \smm)  \nonumber \\
    &\qquad -
    (1-\smm)\log(\empdist\circ P^k)(q|x^{(m)}) \nonumber \\
    &\qquad - \smm\log(\empdist\circ P^k)(q'|x^{(m')})\Big).
\end{align}
Again using Lemma \ref{lem:type_class} (or more simply combining \eqref{eq:sanov1}--\eqref{eq:sanov2}), we have
\begin{align}
    &-(1-\smm ) \log(\empdist\circ P^k)(q|x^{(m)}) \nonumber \\
        &\hspace*{3.5cm}- \smm \log(\empdist\circ P^k)(q'|x^{(m')}) \nonumber \\ 
    & \qquad \geq (1-\smm) k D(q\|P_m) + \smm kD(q'\|P_{m'}).
\end{align}
As such, to prove \eqref{eq:dc_target}, it remains to show that for all $q, q'\in \calP(\calY)$ and $m \ne m'$, we have
\begin{multline}
    \dc(q, q', P^k \circ w, \smm) + (1-\smm)k D(q\|P_m) \\ + \smm kD(q'\|P_{m'}) \geq (k-2)\altrate \label{eq:to_establish}
\end{multline}
This is shown in a series of fairly technical lemmas throughout the rest of the subsection.
	
Let $(c_1, c_2)$ solve the following system of simultaneous equations (the dependence of $c_1, c_2$ on $m,m'$ is left implicit):
\begin{gather}
    \smm  =  c_1 \frac{|\calX|-1}{|\calX|}  + c_2 \frac{1}{|\calX|} \label{eq:c_def1} \\
    1  =  c_1 +  c_2. \label{eq:c_def2}
\end{gather}	
Note that the determinant of this system is non-zero for $|\calX| \ge 3$, so there exists a unique solution.

\begin{lemma} \label{lem:c_bounds}
    The preceding constants $c_1,c_2$ satisfy
    \begin{equation}
       \frac{1-2t}{1-t} \leq c_{\nu} \leq \frac{t}{1-t} \text{~~~for~~~}\nu=1,2,
       \label{eq:c1_c2_bound}
    \end{equation}
    where $t$ is given in \eqref{inequality_t} (with $M = |\calX|$).
\end{lemma}
\begin{proof}
Note that
\begin{align}
    1-t &\geq \smm \\
    &\stackrel{\eqref{eq:c_def1}}{=}  c_1 \frac{|\calX|-1}{|\calX|}  + c_2 \frac{1}{|\calX|} \\
    &\stackrel{\eqref{eq:c_def2}}{=}  c_1 \frac{|\calX|-1}{|\calX|}  + (1-c_1) \frac{1}{|\calX|} \\
    &= c_1 \frac{|\calX|-2}{|\calX|}  + \frac{1}{|\calX|}, \label{eq:c_pf1}
\end{align}
and therefore,
\begin{equation}
    c_1 \stackrel{\eqref{eq:c_pf1}}{\leq} \frac{|\calX|}{|\calX|-2}\left(1-t-\frac{1}{|\calX|}\right) \stackrel{\eqref{inequality_t}}{=} \frac{t}{1-t}.
\end{equation}
In addition, we have
\begin{align}
     t &\leq \smm  \\
     &\stackrel{\eqref{eq:c_def1}}{=}  c_1 \frac{|\calX|-1}{|\calX|}  + c_2 \frac{1}{|\calX|} \\
     &\stackrel{\eqref{eq:c_def2}}{=}  (1-c_2) \frac{|\calX|-1}{|\calX|}  + c_2 \frac{1}{|\calX|} \\
     &= \frac{|\calX|-1}{|\calX|} - c_2 \frac{|\calX|-2}{|\calX|}, \label{eq:c_pf2}
\end{align}
and therefore,
\begin{align}
    c_2 &\stackrel{\eqref{eq:c_pf2}}{\leq} \frac{|\calX|}{|\calX|-2} \left(\frac{|\calX|-1}{|\calX|}-t\right)
    \\ &= \frac{|\calX|}{|\calX|-2}\left(1-t-\frac{1}{|\calX|}\right) \\
    &\stackrel{\eqref{inequality_t}}{=} \frac{t}{1-t}.
\end{align}
This concludes the upper bound in \eqref{eq:c1_c2_bound}, and the lower bound immediately follows since $c_1 + c_2 = 1$ and $1 - \frac{t}{1-t} = \frac{1-2t}{1-t}$.
\end{proof}

For convenience, we repeat \eqref{eq:g_dfn} here:
\begin{equation}
    g(q)_x = \begin{cases}
    \frac1{|\calX|} \min\Big(\frac{D(q\|P_{m_q})}{\altrate}, 1\Big) & x\neq m_q \\
    1-\frac{|\calX|-1}{|\calX|} \min\Big(\frac{D(q\|P_{m_q})}{\altrate}, 1\Big) & x = m_q.
    \end{cases} \label{eq:g_rep}
\end{equation}
where $m_q = \argmin_j D(q\|P_j)$.  Recall also that we are now using $\altrate$ defined in \eqref{eq:altrate} with $t$ given in \eqref{inequality_t}, rather than the previous definition of $\altrate$ corresponding to $t = \frac{1}{2}$.

The following lemma is analogous to Lemma \ref{lem:linear_bound}.

\begin{lemma}  \label{lem:D_ratio}
    We have for all $q \in \calP(\calY)$ and $m \ne m'$ that
    \begin{equation}
        \frac{(1-s_{m,m'}) D(q\|P_m)}{\altrate} \geq c_2(1-g(q)_m) + c_1 g(q)_{m'}. \label{eq:D_ratio}
    \end{equation}
\end{lemma}
\begin{proof}
(\underline{Case 1}: $D(q\|P_j)\geq \altrate$ for all $j$) In this case, \eqref{eq:g_rep} gives $g(q)_m = 1/|\calX|$ for all $m$, so that
\begin{align}
    c_2(1-g(q)_m) + c_1 g(q)_{m'} 
    &= c_2\frac{|\calX|-1}{|\calX|} + c_1 \frac{1}{|\calX|} \\
    & \stackrel{\eqref{eq:c_def1}, \eqref{eq:c_def2}}{=} 1-\smm \\ 
    &\leq \frac{(1-s_{m,m'}) D(q\|P_m)}{\altrate}.
\end{align}

(\underline{Case 2}: $m_q = m$, $D(q\|P_m) < \altrate$) We have
\begin{align}
    &c_2(1-g(q)_m) + c_1 g(q)_{m'} \\ 
    &\stackrel{\eqref{eq:g_rep}}{=} c_2\frac{|\calX|-1}{|\calX|} \frac{D(q\|P_{m})}{\altrate} + c_1 \frac{1}{|\calX|}\frac{D(q\|P_{m})}{\altrate} \\
    & \stackrel{\eqref{eq:c_def1}, \eqref{eq:c_def2}}{=} (1-\smm) \frac{D(q\|P_{m})}{\altrate}.
\end{align}

(\underline{Case 3}: $m_q = m'$, $D(q\|P_{m'}) < \altrate$) We have
\begin{align}
    & c_2(1-g(q)_m) + c_1 g(q)_{m'} \\ &\stackrel{\eqref{eq:g_rep}}{=} c_2\left(1 - \frac{1}{|\calX|} \frac{D(q\|P_{m'})}{\altrate}\right) + c_1\left(1- \frac{|\calX|-1}{|\calX|}\frac{D(q\|P_{m'})}{\altrate}\right)\\
    &= c_2 + c_1 - \frac{D(q\|P_{m'})}{\altrate} \left(c_2\frac{1}{|\calX|} + c_1\frac{|\calX|-1}{|\calX|}\right)\\
    &\stackrel{\eqref{eq:c_def1}}{=} c_2 + c_1 - \frac{D(q\|P_{m'})}{\altrate}\smm \label{eq:simplify}\\
    &\stackrel{\eqref{eq:c_def2}}{=} 1 - \frac{\smm D(q\|P_{m'})}{\altrate}\\
    &\stackrel{Lem.~\ref{minimum_value}}{\le} 1 - \frac{\dc(m,m',P,\smm) - (1-\smm) D(q\|P_{m})}{\altrate}\label{eq:apply_min_value}\\
    &\stackrel{\eqref{eq:altrate}}{\leq} 1 - \frac{\altrate - (1-\smm) D(q\|P_{m})}{\altrate}\\
    &= \frac{(1-\smm) D(q\|P_{m})}{\altrate}.
\end{align}

(\underline{Case 4:} $m_q \notin \{m,m'\}$, $D(q\|P_{m_q}) < \altrate$) 
Since $D(q\|P_m) \geq D(q\|P_{m_q})$, we have
\begin{multline}
    \frac23 D(q\|P_m) + \frac13 D(q\|P_{m_q}) \\ \geq (1- s_{m, m_q}) D(q\|P_m) + s_{m, m_q} D(q\|P_{m_q})\geq  \altrate, \label{eq:to_rearrange}
\end{multline}
where the first inequality uses $s_{m, m_q} \in \big[\frac{1}{3},\frac{2}{3}\big]$, and the second inequality uses Lemma \ref{minimum_value} and the fact that we defined $\smm$ to satisfy $\dc(m, m', P, \smm) \geq \altrate$.
Re-arranging \eqref{eq:to_rearrange}, we obtain
\begin{equation}
    \frac{(1-\smm) D(q\|P_m)}{\altrate} \geq (1-\smm) \left(\frac32 - \frac12 \frac{D(q\|P_{m_q})}{\altrate}\right).
\end{equation}
Moreover, the choice of $g$ in \eqref{eq:g_rep} gives
\begin{equation}
    c_2(1-g(q)_m) + c_1 g(q)_{m'} = c_2 + (c_1-c_2) \frac{1}{|\calX|} \frac{D(q\|m_q)}{\altrate}.
\end{equation}
Thus, it remains to show that
\begin{equation}
    (1-\smm) \left(\frac32 - \frac12 \frac{D(q\|m_q)}{\altrate}\right) \geq c_2 + (c_1-c_2) \frac{1}{|\calX|} \frac{D(q\|m_q)}{\altrate}. \label{eq:two_cases}
\end{equation}
Since this is a linear function in $D(q\|m_q)$ and $0\leq D(q\|m_q) \leq \altrate$, we only need to check the endpoints:
\begin{itemize}
    \item When $D(q\|m_q) = \altrate$, \eqref{eq:two_cases} reduces to $1-\smm \ge c_2 + (c_1 - c_2)\frac{1}{|\calX|}$, which holds with equality in view of \eqref{eq:c_def1} and \eqref{eq:c_def2}.
    \item When $D(q\|m_q)=0$, \eqref{eq:two_cases} reduces to $\frac32 (1-\smm) \geq c_2$, which holds because
    \begin{equation}
        1-\smm \stackrel{\eqref{eq:c_def1},\eqref{eq:c_def2}}{=} \frac{|\calX|-1}{|\calX|}c_2 + \frac{1}{|\calX|}c_1 \geq  \frac{|\calX|-1}{|\calX|}c_2 \geq \frac23 c_2
    \end{equation}
    under our assumption $|\calX|\geq 3$.
\end{itemize}
\end{proof}

\begin{lemma}
We have for all $q,q' \in \calP(\calY)$ and $m \ne m'$ that
\begin{multline}
    (1-\smm) k \cdot D(q||P_m) + \smm k \cdot D(q'||P_{m'}) \\ \geq k \cdot \altrate (1 - c_2(g(q)-g(q'))_m - c_1(g(q')-g(q))_{m'}).
\end{multline}
\label{lem:linear_bound2}
\end{lemma}

\begin{proof}
From \eqref{eq:c_def1}--\eqref{eq:c_def2} and $s_{m,m'} = 1- s_{m',m}$, 
swapping $m$ and $m'$ has the effect of swapping $c_1$ and $c_2$. Hence, a symmetric argument to the proof of Lemma \ref{lem:D_ratio} gives the following analog of \eqref{eq:D_ratio}:
\begin{equation}
    \frac{\smm D(q'||P_{m'})}{\altrate} \geq c_1 (1-g(q')_{m'}) + c_2 g(q')_{m}.  \label{eq:D_ratio2}
\end{equation}
Combining \eqref{eq:D_ratio} and \eqref{eq:D_ratio2} gives
\begin{align}
    &(1-\smm) k \cdot D(q||P_m) + \smm k \cdot D(q'||P_{m'})\\ 
    &\quad\geq k \cdot \altrate (c_2(1-g(q)_m) + c_1 g(q)_{m'} + c_1 (1-g(q')_{m'}) \nonumber \\
        &\hspace*{5.6cm} + c_2 g(q')_{m})\\
    &\quad \stackrel{\eqref{eq:c_def2}}{=} k \cdot \altrate (1 - c_2(g(q)-g(q'))_m - c_1(g(q')-g(q))_{m'}).
\end{align}
\end{proof}

\begin{lemma}
For all $x, x' \in \calX$ and $m \ne m'$, we have
\begin{equation}
    \dc(x, x', P, \smm) \geq
    \begin{cases}
        \altrate & (x,x') = (m,m')\\
        \frac{t}{1-t} \altrate & \text{ otherwise. }
    \end{cases}
\end{equation}
\label{lem:worst_case_bound}
\end{lemma}
\begin{proof}
    The case $(x,x') = (m,m')$ is trivial, as it precisely reduces to how we defined $s_{m,m'}$.  In the following, we focus on the other case.
    
    The function $\dc(x, x', P, s)$ is concave in $s$ (see \cite[Theorem 5]{berlekampI}).  As a result, in the case that $\smm \leq s_{x,x'}$, Jensen's inequality gives
    \begin{align}
        &\dc(x, x', P, \smm) \nonumber \\ 
        &\geq \frac{\smm}{s_{x,x'}} \dc(x, x', P, s_{x,x'}) + \Big(1-\frac{\smm}{s_{x,x'}}\Big) \dc(x, x', P, 0) \\
        &\geq \frac{t}{1-t}\altrate,
    \end{align}
    where we lower bounded the $\dc$ terms by $\altrate$ and $0$ respectively, and used $s \in [t,1-t]$.
    
    On the other hand, for $\smm > s_{x,x'}$, we can use the fact that $\smm = 1-s_{m',m}$ to write $\dc(x, x', P, \smm) = \dc(x', x, P, s_{m',m})$, 
    after which we can apply the same argument as the first case, since $s_{m',m} = 1-\smm \leq 1-s_{x, x'} = s_{x', x}$.
\end{proof}

\begin{lemma}
    For all $q,q' \in \calP(\calY)$ and all $m \ne m'$, we have
    \begin{multline}
    \dc(w^{(q)}, w^{(q')}, P^k, \smm) \geq k \cdot \altrate\Big(c_2 (g(q) - g(q'))_m \\ +  c_1  (g(q') - g(q))_{m'}\Big) - 2 \altrate.
    \label{eq:weighted_bound}
    \end{multline}
    \label{lem:weighted_bound}
\end{lemma}
\begin{proof}
Define $N(x,x')$ to be the number of indices $i \in \{1,\dotsc,k\}$ such that $w^{(q)}_i = x$ and $w^{(q')}_i = x'$.  We have
\begin{align}
    &\dc(w^{(q)}, w^{(q')}, P^k, \smm)  \\
    & = \sum_{x, x'} N(x, x') \dc(x, x', P, \smm) \label{eq:nmm_step1} \\
    & \geq \sum_{x'} N(m, x') \dc(m, x', P, \smm) \\
    & \geq \sum_{x' \neq m} N(m, x') \frac{t}{1-t} \altrate + N(m, m') \frac{1-2t}{1-t} \altrate,
    \label{eq:positive_coefficient_nmm}
\end{align}
where \eqref{eq:nmm_step1} uses the tensorization property (Lemma \ref{lem:tensorize}), and \eqref{eq:positive_coefficient_nmm} uses Lemma \ref{lem:worst_case_bound} and the fact that $\frac{t}{1-t}+\frac{1-2t}{1-t} = 1$.  Moreover, following a similar but modified set of steps, we have
\begin{align}
    & \dc(w^{(q)}, w^{(q')}, P^k, \smm) \nonumber \\
      & = \sum_{x, x'} N(x, x') \dc(x, x', P, \smm) \\
    & \geq \sum_{\substack{(x,x') \,:\, \\ x=m \text{ or } x'=m'}} N(x, x') \dc(m, x', P, \smm) \\
    & \geq \sum_{x' \neq m} N(m, x') \frac{t}{1-t} \altrate + \sum_{x \neq m'} N(x, m') \frac{t}{1-t} \altrate \nonumber \\
        &\hspace*{3.5cm} + N(m, m') \frac{1-3t}{1-t} \altrate.
    \label{eq:negative_coefficient_nmm}
\end{align}
We add $\frac{3t-1}{t}$ times \eqref{eq:positive_coefficient_nmm} together with $\frac{1-2t}{t}$ times \eqref{eq:negative_coefficient_nmm} (note that $\frac13 \leq t \leq \frac12$, so that the weights are non-negative), so that $N(m,m')$ cancels out:
\begin{align}
    &\dc(w^{(q)}, w^{(q')}, P^k, \smm) \nonumber \\ 
    &\geq \sum_{x' \neq m} N(m, x') \frac{t}{1-t} \altrate + \sum_{x \neq m'} N(x, m') \frac{1-2t}{1-t} \altrate, \label{eq:to_simplify}
\end{align}
where the coefficient $\frac{t}{1-t}$ arises by simplifying $\frac{3t-1}{t} \cdot \frac{t}{1-t} + \frac{1-2t}{t} \cdot \frac{t}{1-t}$.

Let $N_q(x)$ denote the number of indices $i \in \{1,\ldots, k\}$ with $w_i^{(q)} = x$, and similarly for $N_{q'}(x)$.  Since $\sum_{x'} N(m,x') = N_q(m)$ and $N(m,m) \le N_{q'}(m)$, we have

\begin{equation}
    \sum_{x' \neq m} N(m,x') \geq N_q(m) - N_{q'}(m) \geq k (g(q) - g(q'))_m - 2,
\end{equation}
with the right-hand side coming from \eqref{eq:rounding_error}.  Similarly, we have
\begin{equation}
    \sum_{x \neq m'} N(m',x) \geq N_{q'}(m') - N_{q}(m')  \geq k (g(q') - g(q))_{m'} - 2,
\end{equation}
so that \eqref{eq:to_simplify} can be weakened to 
\begin{multline}
    \dc(w^{(q)}, w^{(q')}, P^k, \smm) \geq k \cdot \altrate\Big(\frac{t}{1-t} (g(q) - g(q'))_m \\ +  \frac{1-2t}{1-t}  (g(q') - g(q))_{m'}\Big) - 2\altrate.
    \label{eq:weighted_bound_mirror0}
\end{multline}

A similar argument with $q,q'$ interchanged and $m,m'$ interchanged gives
\begin{multline}
    \dc(w^{(q')}, w^{(q)}, P^k, s_{m',m})  
        \geq k \cdot \altrate\Big(\frac{1-2t}{1-t} (g(q) - g(q'))_m \\ +  \frac{t}{1-t}  (g(q') - g(q))_{m'}\Big) - 2\altrate.
    \label{eq:weighted_bound_mirror}
\end{multline}
Since $c_1 + c_2 = 1$ (see \eqref{eq:c_def2}), $ c_1, c_2 \in \big[\frac{1-2t}{1-t},\frac{t}{1-t}\big]$ (see Lemma \ref{lem:c_bounds}), and $\frac{t}{1-t} + \frac{1-2t}{1-t} = 1$, we can combine
\eqref{eq:weighted_bound_mirror0} and \eqref{eq:weighted_bound_mirror} via a suitable convex combination (i.e., add $\lambda$ times one and $1-\lambda$ times the other, where $\lambda \in [0,1]$) to obtain the desired inequality \eqref{eq:weighted_bound}.
\end{proof}

Finally, we combine Lemmas \ref{lem:linear_bound2} and \ref{lem:weighted_bound} to obtain
\begin{align}
    &\dc(w^{(q)}, w^{(q')}, P^k \circ w, \smm) \nonumber \\
        &\quad+ (1-\smm)k D(q\|P_m) + \smm kD(q'\|P_{m'})\\
    &\geq k \cdot \altrate\left(c_2 (g(q) - g(q'))_m +  c_1  (g(q') - g(q))_{m'}\right) - 2 \altrate \nonumber \\
    &\qquad + k \cdot \altrate (1 - c_2(g(q)-g(q'))_m - c_1(g(q')-g(q))_{m'}) \\
    &\geq (k-2) \altrate,
\end{align}
which establishes \eqref{eq:to_establish} and completes the proof of Theorem \ref{thm:almost}.

\subsection{Proof of Theorem \ref{thm:non_reversible_equal_rate}} \label{sec:pf_nonpw_rate}

Let $\edag < \rateone_M$, and fix an $(M,\ell)$-codebook satisfying the conditions of Theorem \ref{thm:berlekamp}.  We use Lemma \ref{lem:tensorize} to write
\begin{equation}
    \dc(x^{(m)}, x^{(m')}, P^\ell, s) = \sum_{i=1}^\ell \dc(x^{(m)}_i, x^{(m')}_i, P^\ell, s),
\end{equation}
and observe that when this quantity is treated as a function of $s$, it is decreasing from $[0,t]$ and decreasing in $[1-t,1]$, due to \eqref{condition_argmin} and the concavity of $\dc$ in $s$ \cite{berlekampI}.  Therefore, this quantity is maximized by some $s \in [t, 1-t]$, so that
\begin{equation}
     \dc(x^{(m)}, x^{(m')}, P^\ell) = \max_{t \leq s \leq 1-t} \dc(x^{(m)}, x^{(m')}, P^\ell, s). \label{eq:dc_nonpw}
\end{equation}
The remainder of the proof is the same as that of Section \ref{sec:pf_fixed}, so we only treat it briefly: Let $P'$ be the restriction of $P^\ell$ to $\calC$, and observe that
\begin{equation}
    \min_{m\neq m'} \max_{t \leq s \leq 1-t} \dc(x^{(m)}, x^{(m')}, P^\ell, s) \geq \ell \cdot \edag, \label{eq:min_nonpw}
\end{equation}
where the inequality combines \eqref{eq:dc_nonpw} with the conclusion of Theorem \ref{thm:berlekamp}.

Applying Theorem \ref{thm:almost} to \eqref{eq:min_nonpw} readily gives $\ratetwo_{M,P} \geq \frac{1}{\ell} \ratetwo_{M,P'} = \edag$.  Since $\edag$ is arbitrarily close to $\rateone_M$, this completes the proof.
    
\section{Zero-Rate Error Exponents (Proof of of Theorem \ref{thm:zero_rate})}

In this section, we prove Theorem \ref{thm:zero_rate}, which concerns the limiting error exponent as the coding rate tends to zero from above.  While this corresponds to a number of messages growing with the block length, the following result concerning a fixed number of messages will serve as a useful building block.  This result can be deduced from classical works such as \cite{berlekamp,jelinek1968eval}, but we also provide a short proof in Appendix \ref{sec:pf_expurg}.

	\begin{lemma}
    	For any $\otherrate < \erateone(0)$ and any positive integer $M$, it holds for sufficiently large $\ell$ that there exists an $(M,\ell)$-codebook $\calC'$ such that 
    	\begin{equation}
    	\dbmin(\calC', P^\ell) \geq \ell \cdot \otherrate.
    	\end{equation}
    	\label{lem:expurgation}
	\end{lemma}
    
    \begin{lemma}
    For any discrete memoryless channel $P$ with input alphabet $\calX$, we have
    \begin{equation}
        \erateone_P(0) \geq \frac{|\calX|-1}{|\calX|} \dbmin(\calX, P). \label{eq:E1LB}
    \end{equation}
    \label{lem:zero_rate_uniform}
    \end{lemma}
    \begin{proof}
        We apply Theorem \ref{zero_rate} and lower bound \eqref{eq:E1zero} by substituting the uniform distribution for $q$.  Then $\db(x,x',P)$ is zero when $x = x'$, and is at least $\dbmin(\calX, P)$ when $x \ne x'$, which yields \eqref{eq:E1LB}. 
    \end{proof}
    
    Fix any $\otherrate < \erateone(0)$ and any finite $M$, and let $\calC'$ be the $(M,\ell)$-codebook in Lemma \ref{lem:expurgation} with some large enough $\ell$.  Moreover, let $P'$ be the restriction of $P^\ell$ on $\calC'$.
    
    For any positive integer $k$, any function $f: \calY_{P'}^k \rightarrow \calX_{P'}^k$ and any rate $R>0$,\footnote{Note that the finite constant $M$ should not be confused with the growing number of messages here.} we can follow the same argument as the one leading to \eqref{eq:E2_to_E1}:
    \begin{align}
        &\eratetwo_{P'}(R) \nonumber \\
        &= \liminf_{n'\rightarrow\infty} -\frac1{n'}\log P^{(2)}_e(n',e^{n' R}, P')\\
        & \stackrel{\eqref{eq:two_to_one}}{\geq} \liminf_{n' \rightarrow\infty} -\frac{1}{n'}\log P^{(1)}_e(\lfloor n'/k-1\rfloor,e^{n'R}, (P')^k \circ f \circ (P')^k)\\
        & = \frac1k \erateone_{(P')^k \circ f \circ (P')^k}(R).
    \end{align}
    Taking $R \rightarrow 0^+$, we obtain
    \begin{align}
        \eratetwo_{P'}(0) &\geq \frac1k \erateone_{(P')^k \circ f \circ (P')^k}(0) \\ 
        &\geq \frac1k \frac{M-1}{M} \dbmin(\calC', (P')^k \circ f \circ (P')^k),
        \label{eq:rate_limit}
    \end{align}
    where the last inequality comes from Lemma \ref{lem:zero_rate_uniform}.  We note that here, in contrast to \eqref{eq:E2_to_E2_weakened}, we do not consider the use of repetition at the encoder.  Instead, Lemma \ref{lem:zero_rate_uniform} is based on the expurgated exponent of the channel $(P')^k \circ f \circ (P')^k$, which implicitly requires complicated coding over the large super-alphabet.
    
    Applying Lemma \ref{lem:db_sep} to the channel $P'$, and noting that $\altrate$ in \eqref{eq:shorthand} reduces to $\dbmin(\calC', P')$, we deduce that there exists a suitable function $f$ such that 
    \begin{multline}
        \dbmin(\calC', (P')^k \circ f \circ (P')^k) \\ \geq (k-2)\dbmin(\calC', P') - 2|\calY_{P'}|\log(k+1).
        \label{eq:max_f}
    \end{multline}
    By the preceding application of Lemma \ref{lem:expurgation}, we have $\dbmin(\calC', P') \geq \ell \cdot \otherrate$ with $\otherrate$ being arbitrarily close to $\erateone(0)$.  Moreover, since $\ell$ is constant, the term $|\calY_{P'}|$ in \eqref{eq:max_f} is also constant (albeit possibly large).
    
    Substituting \eqref{eq:max_f} into \eqref{eq:rate_limit} and taking $k\rightarrow \infty$ (note that our choice of $P'$ does not depend on $k$) gives
    \begin{equation}
        \eratetwo_{P'}(0) \geq \frac{M-1}{M} \ell \cdot \otherrate.
        \label{eq:ratetwo_final}
    \end{equation}
    To convert this exponent for $P'$ to one for $P$, we can allow the agents to read and write $\ell$ symbols at a time, recalling that $P'$ is a restriction of $P^\ell$.  Analogous to \eqref{eq:super_symbol}, we obtain
    \begin{align}
        P^{(2)}_e(n,e^{nR},P) 
        &\leq P^{(2)}_e(\lfloor n/\ell \rfloor,e^{nR},P^\ell) \\
        &\leq P^{(2)}_e(\lfloor n/\ell \rfloor,e^{nR},P'),
    \end{align}
    and hence,
    \begin{align}
        \eratetwo_P(R) &\geq \liminf_{n\rightarrow\infty} -\frac1n \log P^{(2)}_e(\lfloor n/\ell \rfloor,e^{nR},P') \\
        &= \frac{1}{\ell}\eratetwo_{P'}(\ell \cdot R).
    \end{align}
    Finally, we take $R\rightarrow 0^+$ and combine with \eqref{eq:ratetwo_final} to obtain
    \begin{equation}
        \eratetwo_P(0) \geq \frac{1}{\ell} \eratetwo_{P'}(0) \geq \frac{M-1}{M} \otherrate.
    \end{equation}
    This completes the proof, since $\otherrate$ is arbitrarily close to $\erateone(0)$ and $M$ is arbitrarily large.

\section{Further Extensions}
    
\subsection{Distinct Channels Setting} \label{sec:diff_chan}

Throughout the paper, we have assumed that the two channels have the same transition law.  We now drop this assumption, letting $P$ be the first channel (encoder to relay) and $Q$ be the second channel (relay to decoder).  We denote the corresponding error exponents by $\ratetwo_{M,P,Q}$, $\eratetwo_{P,Q}$, and so on.

Starting with the fixed-$M$ setting, by data processing inequalities, we readily obtain the following converse bound:
\begin{equation}
    \ratetwo_{M,P,Q} \leq \min(\rateone_{M,P}, \rateone_{M,Q}).
\end{equation}
The following result gives a matching achievability bound for pairwise reversible channels, and generalizes Theorem \ref{thm:fixed_alphabet}.

\begin{theorem}
If both $P$ and $Q$ are pairwise reversible, then $\ratetwo_{M,P,Q} = \min(\rateone_{M,P}, \rateone_{M,Q})$.
\end{theorem}

The proof is largely the same as Theorem \ref{thm:fixed_alphabet}, so we only briefly outline some of the differences.  As before, we start by assuming $|\calX_P| = |\calX_Q| = M$, and Lemma \ref{lem:block} is now replaced by
\begin{equation}
    \ratetwo_{M,P,Q} \geq \frac1k \dcmin(\calC, Q^k \circ f \circ P^k).
    \label{eq:block_two}
\end{equation}
The function $f$ (and the auxiliary function $g$ that it depends on) remains the same, except that we generalize $\altrate$ (see \eqref{eq:shorthand}) to
\begin{equation}
    \altrate = \min_{m\neq m'} \min\big( \db(m,m',P), \db(m,m',Q) \big).
\end{equation}
For any $\otherrate < \ratetwo_{M,P,Q}$, we find two $(M,\ell)$-codebooks $\calC_P$ and $\calC_Q$ satisfying the conditions of Theorem \ref{thm:berlekamp}. We then let $P'$ and $Q'$ be the restriction of $P^\ell$ and $Q^\ell$ on the corresponding codebooks, and apply \eqref{eq:block_two} to $(P',Q')$.

By similar reasoning, we also have the following analog of Theorem \ref{thm:zero_rate} for the zero-rate exponent, without any requirement of $P$ or $Q$ being pairwise reversible.

\begin{theorem}
    For any DMCs $P$ and $Q$, we have $\eratetwo_{P,Q}(0) = \min(\erateone_{M,P}(0), \erateone_{M,Q}(0))$.
\end{theorem}

In the fixed-$M$ setting without pairwise reversibility, the situation becomes more complicated; even in the case of $M=2$ studied in \cite{teachlearn}, we do not have a complete solution to the question of when it holds that $\ratetwo_{2,P,Q} = \min(\rateone_{2,P}, \rateone_{2,Q})$ (both positive and negative cases are known, with $P=Q$ being a notable positive case).  Hence, we leave this more challenging setting for future work.

\subsection{A Case Where the 1-Hop and 2-Hop Exponents Differ (Proof of Theorem \ref{thm:mismatch})} \label{sec:differ}

We now return to the case that the two channels and identical, and address the question of how general the result $\ratetwo_{M,P} = \rateone_{M,P}$ might be.  One might hope that this result can be extended to \emph{all} $(M,P)$ with no additional assumptions.  However, here we show that such a level of generality is not possible, thereby proving Theorem \ref{thm:mismatch}.

Fix $p>0$, and let $P$ have the following transition law (with $|\calX|=3$ and $|\calY| = 4$):
\begin{equation}
\begin{bmatrix}
    1-2p & p & 0 & p \\
    0 & 1-2p & p & p \\
    p & 0 & 1-2p & p \\
\end{bmatrix}.
\end{equation}
We will show that for all sufficiently small $p$, it holds that $\ratetwo_{3,P} < \rateone_{3,P}$.   
We number the inputs as $1,2,3$ and the outputs as $1,2,3,\erasure$, where $\erasure$ stands for ``erasure''.

We first claim that for $p < \frac{1}{3}$, we have
\begin{equation}
    \dcmin(\{1,2,3\}, P) = \log\frac{1}{2p}.
\end{equation}
This follows by a direct substitution into the definition of $\dc$; the relevant $s$-dependent expression is $-\log\big( (1-2p)^s p^{1-s} + p \big)$, and for $p < \frac{1}{3}$ we have $1-2p > p$, so that the maximum is attained at $s=0$.  Applying Lemma \ref{lem:initial}, it follows that $\rateone_3 \geq \log\frac{1}{2p}$. 

We will show that whenever $n$ is a multiple of 19, we have
\begin{equation}
    \mathbb{P}_e^{(2)}(n,3) \geq \frac13 p^{18n/19}(1-2p)^{n}. \label{eq:two_hop_bound}
\end{equation}
Note that we can always make the problem easier by rounding up to the next multiple of 19.  Thus, \eqref{eq:two_hop_bound} implies that we can upper bound the 2-hop exponent by $\frac{18}{19}\log\frac{1}{p} + \log\frac{1}{1-2p}$.  For sufficiently small $p$, this is strictly smaller than $\log\frac{1}{2p}$, yielding the desired claim $\ratetwo_3 < \rateone_3$.

To prove \eqref{eq:two_hop_bound}, we consider a relaxed version of the two-hop problem.  We split the transmission time into blocks of length $6n/19$ and $13n/19$ respectively, and consider the following setup:
\begin{itemize}
    \item In the first block, the encoder sends $6n/19$ symbols to the relay (via $P$).
    \item If only the erasure symbol $\erasure$ is received by the relay in the first block (which occurs with probability $p^{6n/19}$ regardless of $\Theta$), then:
    \begin{itemize}
        \item The relay sends $6n/19$ symbols to the decoder in the first block; call this string $w'_{\erasure}$.
        \item The relay then learns the true value of $\Theta$, and then sends another $13n/19$ symbols to the decoder; call this string $w^{\dagger\dagger}_{\Theta}$.
    \end{itemize}
    \item Otherwise, if the relay receives any symbols among $\{1,2,3\}$ in the first block (i.e., non-erasures), then:
    \begin{itemize}
        \item The relay learns the true value of $\Theta$ immediately, and sends $n$ symbols to the decoder. Let $w'_{\Theta}$ contain the first $6n/19$ symbols, and $w''_{\Theta}$ contain the remaining $13n/19$ symbols.
    \end{itemize}
\end{itemize}
Observe that this is an easier problem than the original one, because the relay is either given $\Theta$ ``for free'' for the entire transmission time, or it is given $\Theta$ after the first block in a scenario where the first block it received gave no information about $\Theta$ anyway (i.e., all erasures).  Thus, any converse in this setting implies a converse in the original setting.  We used similar ideas for the case $M=2$ in \cite[Sec.~III-D]{teachlearn}.


We will show that under this modified setting, the error probability is at least $\frac13 p^{18n/19}(1-2p)^n$.  We proceed with a proof by contradiction, instead assuming that
\begin{equation}
    \p({\rm error}) < \frac13 p^{18n/19}(1-2p)^n. \label{eq:assumed_err}
\end{equation}
In the following, we use the notation $u \oplus v$ to denote the concatenation of two strings.

We first claim that $w'_{1}$ and $w'_{2}$ can only share the same symbol in fewer than $n/19$ positions.  To see this, suppose that they agree in $n/19$ or more positions, and consider the event that (i) at least one symbol from the encoder is non-erased in the first block, and (ii) all symbols from the relay for which $w'_{1} \oplus w''_{1}$ and $w'_{2} \oplus w''_{2}$ differ are erased.  Conditioned on either $\Theta=1$ or $\Theta=2$, the probability of this occurring is lower bounded by $p^{18n/19}(1-2p)^n$,\footnote{The probability of having at least one non-erasure in the first encoder block is lower bounded by $1-p \ge 1-2p$, and this is factored into the $(1-2p)^{n}$ term in which the exponent of $n$ is a crude upper bound on the actual power of $1-2p$. \label{foot:crude}} and when it occurs, the decoder has no information for distinguishing between these two $\Theta$ values.  Since $\p(\Theta=1) = \p(\Theta=2) = \frac{1}{3}$, we deduce that the error probability is at least $\frac13 p^{18n/19}(1-2p)^n$, contradicting \eqref{eq:assumed_err}.  

By the same argument, we can assume that among $w'_{1}$, $w'_{2}$, and $w'_{3}$, any two strings share the same symbol in fewer than $n/19$ positions, and we proceed under this assumption.

Now consider the strings $w'_{1}, w'_{2}, w'_{3}, w'_{\erasure}$.  Since $|\calX| = 3$, in each of the $6n/19$ positions, at least two of them share the same symbol. Therefore, by summing the number of common symbols over all $\binom{4}{2}=6$ pairs, we conclude that at least one pair contains the same symbol in at least $n/19$ positions. From the preceding paragraph, one of them must be $w'_{\erasure}$, and without loss of generality, we can let the other one be $w'_{1}$.  Thus,
\begin{equation}
    \text{(\#symbols in common between $w'_{1}$ and $w'_{\erasure}$)} \ge \frac{n}{19}. \label{eq:common_symb}
\end{equation}

Without loss of generality, assume that $w'_{1}$ and $w''_{1}$ consist of only the symbol 1; if not, we can perform suitable cyclic shifts (i.e., $1 \rightarrow 2 \rightarrow 3 \rightarrow 1$) on every codeword symbol-by-symbol, and suitably apply the inverse shifts at the decoder. 
Then, consider the event $\calA$ described by the following two conditions:
\begin{enumerate}
    \item One of the following two events occurs:
    \begin{itemize}
        \item It holds that $\Theta = 1$, and at least one non-erasure occurs in the first encoder block, so that the relay transmits $w'_{1} \oplus w''_{1}$.
        \item It holds that $\Theta = 2$, and the first encoder block is all erased, so that the relay sends $w'_{\erasure} \oplus w^{\dagger\dagger}_{2}$.
    \end{itemize}
    \item For each position in $\{1,\dotsc,n\}$, depending on the symbol at the corresponding position of $w'_{\erasure} \oplus w^{\dagger\dagger}_{2}$, we have the following:
    \begin{itemize}
        \item If the symbol is 1 (respectively, 2), the decoder receives 1 (respectively, 2).
        \item If the symbol is 3, the decoder receives $\erasure$.
    \end{itemize}
    Note that although this event is defined with respect to $w^{\dagger\dagger}_{2}$, we require this condition to hold regardless of whether $\Theta=1$ or $\Theta=2$.
\end{enumerate}

First consider conditioning on $\Theta = 1$.  We know from \eqref{eq:common_symb} that $w'_{\erasure}$ contains at least $n/19$ 1s, and for $\calA$ to occur, we require that (i) the sent 1s in those corresponding positions are received as 1s, and (ii) at the locations where $w'_{\erasure} \oplus w^{\dagger\dagger}_{2}$ is in $\{2,3\}$, a $1 \to 2$ or $1 \to \erasure$ transition occurs.  Since the latter transitions both occur with probability $p$, we deduce that the probability of $\calA$ occurring given $\Theta=1$ is lower bounded by $p^{18n/19}(1-2p)^{n}$.  (See also Footnote \ref{foot:crude} regarding the requirement of the first encoder block not being entirely erased.)

Now we consider the probability of $\calA$ given $\Theta=2$.  The first block from the encoder is entirely erased with probability $p^{6n/19}$, and given that this is true, the conditional probability of $\calA$ equals $p^{n_3}(1-2p)^{13n/19 - n_3}$, where $n_3$ is the number of 3s in $w^{\dagger\dagger}_{2}$.  Hence, given $\Theta=2$, the probability of $\calA$ is  $p^{6n/19 + n_3} (1-2p)^{13n/19 - n_3} \ge p^{6n/19 + n_3} (1-2p)^n$.

Under event $\calA$, the decoder has no information for distinguishing between $\Theta = 1$ and $\Theta = 2$.  Since $\p(\Theta=1) = \p(\Theta=2) = \frac{1}{3}$, we find that to avoid the preceding lower bounds contradicting the assumed upper bound \eqref{eq:assumed_err}, we require that $p^{6n/19+n_3} \leq p^{18n/19}$, or equivalently, $n_3 \geq 12n/19$. Since there are at least $n/19$ 1s in $w'_{\erasure}$ (see \eqref{eq:common_symb}), there can only be at most $5n/19$ 3s in $w'_{\erasure}$, so there must be at least $7n/19$ 3s in $w^{\dagger\dagger}_{2}$.

By the same argument with $\Theta \in \{1,3\}$ instead of $\{1,2\}$, there are also at least $7n/19$ 3s in $w^{\dagger\dagger}_{3}$. Since $w^{\dagger\dagger}_{2}$ and $w^{\dagger\dagger}_{3}$ have a common length of $13n/19$, we conclude that 
\begin{equation}
    \text{(\#3s in common between $w^{\dagger\dagger}_{2}$ and $w^{\dagger\dagger}_{3}$)} \ge \frac{n}{19}. \label{eq:common_symb2}
\end{equation}
We will finally show that this contradicts \eqref{eq:assumed_err}, roughly via a contrapositive argument to how we deduced \eqref{eq:common_symb} from \eqref{eq:assumed_err}.  
To do so, let $\calB$ be the event that all of the following occur:
\begin{itemize}
    \item The relay first receives $6n/19$ erasures in the first block, and accordingly sends $w'_{\erasure}$.
    \item The relay is then given $\Theta \in \{2,3\}$, and accordingly sends $w^{\dagger\dagger}_{2}$ or $w^{\dagger\dagger}_{3}$.
    \item In the second block, the receiver receives 3 at the positions where $w^{\dagger\dagger}_{2}$ and $w^{\dagger\dagger}_{3}$ equals 3, but receives $\erasure$ at all remaining positions.
\end{itemize}
Regardless of whether we condition on $\Theta=2$ or $\Theta=3$, a similar calculation to that above (but now using \eqref{eq:common_symb2}) yields that $\calB$ occurs with probability at least $p^{18n/19}(1-2p)^n$.  Since the decoder has no information for distinguishing between $\Theta \in \{2,3\}$ when $\calB$ occurs, we obtain the desired contradiction of \eqref{eq:assumed_err}.  This completes the proof of Theorem \ref{thm:mismatch}.

\section{Conclusion}

In this paper, we studied the problem of relaying multiple (or many) bits over a tandem of channels.  
We demonstrated that $\rateone_{M,P} = \ratetwo_{M,P}$ for all pairwise reversible channels, as well as certain channels that are ``almost'' pairwise reversible.  In addition, we showed that $\erateone_P(0) = \eratetwo_P(0)$ for all DMCs regardless of pairwise reversibility, while also showing that there exist channels such that $\ratetwo_{3,P} < \rateone_{3,P}$.  Finally, we generalized our main findings to the case that the two channels differ.

We conclude by briefly raising some interesting open problems:
\begin{itemize}
    \item Under what conditions beyond pairwise reversibility (and its extension in Section \ref{sec:non_pw}) does $\rateone_{M,P} = \ratetwo_{M,P}$?
    \item Under what conditions (if any) do we have $\erateone_P(R) = \eratetwo_P(R) > 0$ when $R > 0$? 
    \item Are fundamentally different protocols needed to attain the optimal 2-hop exponent when $\ratetwo_{M,P} < \rateone_{M,P}$?
    \item What is the smallest possible value of $\frac{\ratetwo_{M,P}}{\rateone_{M,P}}$? The proof of Theorem \ref{thm:mismatch} shows that this ratio can be made arbitrarily close to $\frac{18}{19}$, whereas a simple lower bound is $\frac{1}{2}$.
\end{itemize}

\appendix

\subsection{Proof of Lemma \ref{lem:tensorize}}  \label{sec:pf_tensorize}

We have
\begin{align}
    &\sum_{\vec{y} \in \calY_{P}^k} P^k(\vec{y}|\vec{x})^{1-s} P^k(\vec{y}|\vec{x}')^{s} \\
    &= \sum_{y_1 \in \calY_P} \cdots \sum_{y_k \in \calY_P} \left(\prod_{i=1}^k P(y_i|x_i)^{1-s}\right) \left( \prod_{i=1}^k P^k(y_i|x'_i)^{s} \right)\\
    &= \sum_{y_1 \in \calY_P} P(y_1|x_1)^{1-s} P^k(y_1|x'_1)^{s} \cdots \nonumber \\
        &\hspace*{2cm}\times\sum_{y_k \in \calY_P} P(y_k|x_k)^{1-s} P^k(y_k|x'_k)^{s}\\
    &= \prod_{i=1}^k  \left(\sum_{y_i \in \calY_P} P(y_i|x_i)^{1-s} P^k(y_i|x'_i)^{s}\right)
\end{align}
Taking the negative log on both sides gives \eqref{eq:iid1}, and maximizing over $0\leq s \leq 1$ gives \eqref{eq:iid2}.

\subsection{Proof of Lemma \ref{lem:initial}} \label{sec:pf_initial}

Fix $S = \{x^{(1)}, \ldots, x^{(m)}\}$ with each $x^{(i)} \in \calX_P$. Upon receiving $\Theta$, let the encoder simply send $x^{(\Theta)}$ repeatedly. For each pair of distinct $(x,x')$, we have
\begin{align}
    \p(\hat{\Theta}=x' | \Theta=x) &\leq \exp(-n \cdot \dc(x, x', P)) \\
    &\leq \exp(-n \cdot \dcmin(S,P)),
\end{align}
where the first inequality is a standard Chernoff-style upper bound, and the second inequality follows from the definition of $\dcmin$.  Since there are $M-1$ incorrect values of $\Theta$, a union bound gives
\begin{equation}
    \p(\hat{\Theta} \neq \Theta)  \leq (M-1) \exp(-n \cdot \dcmin(S,P)),
\end{equation}
and since $M$ is fixed (not scaling with $n$), this implies
\begin{align}
    \rateone_M &\leq \liminf_{n\rightarrow \infty} -\frac1n \log\Big((M-1) \exp(-n \cdot \dcmin(S,P))\Big) \\ &= \dcmin(S,P).
\end{align}

\subsection{Proof of Lemma \ref{lem:expurgation}}  \label{sec:pf_expurg}

Let $q \in \calP(\calX)$ be the distribution achieving the maximum in the expression for $E^{(1)}(0)$ in \eqref{eq:E1zero}.  Moreover, let $x$ and $x'$ be two independently chosen codewords of length $\ell$, where each symbol is chosen independently with distribution $q$.  Define $\hat{q}_{i_1, i_2}(x, x')$ to be the fraction of positions such that symbol $i_1$ occurs in $x$ and symbol $i_2$ occurs in $x'$. 

Since $x$ and $x'$ are independent, we have $\e[\hat{q}_{i_1, i_2}(x,x')] = q_{i_1}q_{i_2}$.  Hence, and by the tensorization property of $\db$ (Lemma \ref{lem:tensorize}), we have
\begin{align}
    \frac1\ell \e[\db(x, x',P^\ell)] &= \e\left[\sum_{i_1, i_2} \db(i_1, i_2, P)q_{i_1,i_2}(x, x')\right] \\ 
    &= E^{(1)}(0) > \otherrate.
\end{align}
Since $\frac1\ell \db(x, x',P^\ell)$ is the average of $\ell$ independent variables and $M$ is constant, by the law of large numbers, it holds for sufficiently large $\ell$ that
\begin{equation}
    \p(\db(x, x', P^\ell) \leq \ell \cdot \otherrate) \leq \frac{1}{2M^2}.
    \label{lln_s}
\end{equation}
Now consider choosing $M$ independent codewords according to the i.i.d.~distribution on $q$.  By using a union bound over all $\binom{M}{2} \le M^2$ codeword pairs in \eqref{lln_s}, we deduce that there exists an $(M,\ell)$-codebook $\calC'$ such that $\dbmin(\calC', P^\ell) \geq \ell \cdot \otherrate$, as required.

\bibliographystyle{IEEEtran}
\bibliography{general}    

\begin{thebibliography}{10}
\providecommand{\url}[1]{#1}
\csname url@samestyle\endcsname
\providecommand{\newblock}{\relax}
\providecommand{\bibinfo}[2]{#2}
\providecommand{\BIBentrySTDinterwordspacing}{\spaceskip=0pt\relax}
\providecommand{\BIBentryALTinterwordstretchfactor}{4}
\providecommand{\BIBentryALTinterwordspacing}{\spaceskip=\fontdimen2\font plus
\BIBentryALTinterwordstretchfactor\fontdimen3\font minus
  \fontdimen4\font\relax}
\providecommand{\BIBforeignlanguage}[2]{{%
\expandafter\ifx\csname l@#1\endcsname\relax
\typeout{** WARNING: IEEEtran.bst: No hyphenation pattern has been}%
\typeout{** loaded for the language `#1'. Using the pattern for}%
\typeout{** the default language instead.}%
\else
\language=\csname l@#1\endcsname
\fi
#2}}
\providecommand{\BIBdecl}{\relax}
\BIBdecl

\bibitem{onebit}
W.~Huleihel, Y.~Polyanskiy, and O.~Shayevitz, ``Relaying one bit across a
  tandem of binary-symmetric channels,'' \emph{IEEE International Symposium on
  Information Theory (ISIT)}, 2019.

\bibitem{jog2020teaching}
V.~{Jog} and P.~L. {Loh}, ``Teaching and learning in uncertainty,'' \emph{IEEE
  Transactions on Information Theory}, vol.~67, no.~1, pp. 598--615, 2021.

\bibitem{teachlearn}
Y.~H. Ling and J.~Scarlett, ``Optimal rates of teaching and learning under
  uncertainty,'' \emph{IEEE Transactions on Information Theory}, vol.~61,
  no.~11, pp. 7067--7080, 2021.

\bibitem{gallager}
R.~Gallager, \emph{Information Theory and Reliable Communication}.\hskip 1em
  plus 0.5em minus 0.4em\relax John Wiley \& Sons, Inc., 1968.

\bibitem{berlekamp}
C.~Shannon, R.~Gallager, and E.~Berlekamp, ``Lower bounds to error probability
  for coding on discrete memoryless channels. ii,'' \emph{Information and
  Control}, vol.~10, no.~5, p. 522–552, 1967.

\bibitem{schulman_1994}
S.~Rajagopalan and L.~Schulman, ``A coding theorem for distributed
  computation,'' \emph{ACM Symposium on Theory of Computing}, 1994.

\bibitem{highrates}
V.~Y.~F. Tan, ``On the reliability function of the discrete memoryless relay
  channel,'' \emph{IEEE Transactions on Information Theory}, vol.~61, no.~4,
  pp. 1550--1573, 2015.

\bibitem{bradford2012error}
G.~J. Bradford and J.~N. Laneman, ``Error exponents for block {M}arkov
  superposition encoding with varying decoding latency,'' in \emph{IEEE
  Information Theory Workshop}.\hskip 1em plus 0.5em minus 0.4em\relax IEEE,
  2012, pp. 237--241.

\bibitem{multihopping}
W.~Zhang and U.~Mitra, ``Multihopping strategies: An error-exponent
  comparison,'' in \emph{IEEE International Symposium on Information Theory},
  2007.

\bibitem{concat}
G.~D. Forney, \emph{Concatenated Codes}.\hskip 1em plus 0.5em minus 0.4em\relax
  MIT Press, 1965.

\bibitem{endotend}
A.~Chaaban and A.~Sezgin, ``Multi-hop relaying: An end-to-end delay analysis,''
  \emph{IEEE Transactions on Wireless Communications}, vol.~15, no.~4, pp.
  2552--2561, 2016.

\bibitem{gamalkim}
A.~El~Gamal and Y.-H. Kim, \emph{Network information theory}.\hskip 1em plus
  0.5em minus 0.4em\relax Cambridge University Press, 2011.

\bibitem{fong2017achievable}
S.~L. Fong and V.~Y. Tan, ``Achievable rates for {G}aussian degraded relay
  channels with non-vanishing error probabilities,'' \emph{IEEE Transactions on
  Information Theory}, vol.~63, no.~7, pp. 4183--4201, 2017.

\bibitem{hyptestrelay}
S.~Salehkalaibar, M.~Wigger, and L.~Wang, ``Hypothesis testing over the two-hop
  relay network,'' \emph{IEEE Transactions on Information Theory}, vol.~65,
  no.~7, pp. 4411--4433, 2019.

\bibitem{cover_thomas}
T.~M. Cover and J.~A. Thomas, \emph{Elements of information theory}.\hskip 1em
  plus 0.5em minus 0.4em\relax John Wiley \& Sons, Inc., 2006.

\bibitem{ck}
I.~Csisz\'{a}r and J.~K\"{o}rner, \emph{Information Theory: Coding Theorems for
  Discrete Memoryless Systems}, 2nd~ed.\hskip 1em plus 0.5em minus 0.4em\relax
  Cambridge University Press, 2011.

\bibitem{berlekampI}
C.~Shannon, R.~Gallager, and E.~Berlekamp, ``Lower bounds to error probability
  for coding on discrete memoryless channels. i,'' \emph{Information and
  Control}, vol.~10, no.~1, p. 65–103, 1967.

\bibitem{jelinek1968eval}
F.~Jelinek, ``Evaluation of expurgated bound exponents,'' \emph{IEEE
  Transactions on Information Theory}, vol.~14, no.~3, pp. 501--505, 1968.

\end{thebibliography}

     \begin{IEEEbiographynophoto}{Yan Hao Ling}
        received the B.Comp.~degree in computer science and the B.Sci.~degree 
        in mathematics from the National University of Singapore (NUS) in 2021. 
        He is now a PhD student in the Department of Computer Science at NUS.
        His research interests are in the areas of
        information theory, statistical learning, and theoretical computer science.
    \end{IEEEbiographynophoto}
    
     \begin{IEEEbiographynophoto}{Jonathan Scarlett}
        (S'14 -- M'15) received 
        the B.Eng.~degree in electrical engineering and the B.Sci.~degree in 
        computer science from the University of Melbourne, Australia. 
        From October 2011 to August 2014, he
        was a Ph.D. student in the Signal Processing and Communications Group
        at the University of Cambridge, United Kingdom. From September 2014 to
        September 2017, he was post-doctoral researcher with the Laboratory for
        Information and Inference Systems at the \'Ecole Polytechnique F\'ed\'erale
        de Lausanne, Switzerland. Since January 2018, he has been an assistant
        professor in the Department of Computer Science and Department of Mathematics,
        National University of Singapore. His research interests are in
        the areas of information theory, machine learning, signal processing, and
        high-dimensional statistics. He received the Singapore National Research Foundation (NRF)
        fellowship, and the NUS Presidential Young Professorship award.
    \end{IEEEbiographynophoto}
    
\end{document}